\documentclass[aps,prl,10pt,superscriptaddress,twocolumn,preprintnumbers,floatfix,nofootinbib, nobalancelastpage]{revtex4-2}
\usepackage{mathtools}
\usepackage{amsthm}
\usepackage{amsfonts,amssymb }
\usepackage{graphicx} 
\usepackage{physics}
\usepackage{xcolor}
\usepackage{dsfont}
\usepackage{bm}
\usepackage{comment}
\usepackage{svg}
\usepackage{dsfont}
\newtheorem{theorem}{Theorem}
\newtheorem{observation}{Observation}

\usepackage[hidelinks]{hyperref}

\theoremstyle{definition}

\renewcommand{\selectlanguage}[1]{}
\begin{document}
\title{Quantum thermalization achieves optimal approximate quantum error correction}
\author{Aditi Venkatesh}
\thanks{Electronic address: avenkatesh@g.harvard.edu}
\affiliation{Department of Physics, Harvard University, Cambridge, MA, 02138}
\affiliation{Center for Theoretical Physics --- a Leinweber Institute, Massachusetts Institute of Technology, Cambridge, MA, 02139}

\author{Richard R. Allen}
\affiliation{Center for Theoretical Physics --- a Leinweber Institute, Massachusetts Institute of Technology, Cambridge, MA, 02139}

\author{Saúl Pilatowsky-Cameo}
\affiliation{Center for Theoretical Physics --- a Leinweber Institute, Massachusetts Institute of Technology, Cambridge, MA, 02139}

\author{Bingtian Ye}
\affiliation{Center for Theoretical Physics --- a Leinweber Institute, Massachusetts Institute of Technology, Cambridge, MA, 02139}

\author{Soonwon Choi}
\affiliation{Center for Theoretical Physics --- a Leinweber Institute, Massachusetts Institute of Technology, Cambridge, MA, 02139}
\preprint{MIT-CTP/6104}

\begin{abstract}
    Quantum thermalization explains how an isolated many-body system naturally evolves towards a thermal state, rendering information about the initial conditions inaccessible to local measurements. This is precisely the mechanism utilized in quantum error correction, where information is protected by design through a nonlocal encoding. In this work, we leverage this connection to port the rigorous framework of (approximate) quantum error correction to the study of quantum thermalization. Treating typical late-time states as codewords, we characterize the error-correcting properties of generic thermalizing dynamics. We numerically uncover a universal relationship between the encoding rate, distance, and thermal entropy density of the emergent code. At infinite temperature, this universal curve saturates the quantum Singleton bound, achieving the same optimal limit as Haar-random codes. At finite temperature, we introduce a code family based on the Scrooge ensemble, the natural thermal analogue of the Haar ensemble, and prove it saturates the entropic quantum Singleton bound, establishing this family as optimal within entropic constraints. Our extracted universal curve independently saturates this same bound, revealing that finite-temperature thermalization is itself optimal. Finally, we show how conserved quantities limit the error-correcting behavior of thermalization: codewords with differing energies, or other conserved charges, leak only classical information, and correctability persists until the difference reaches the scale of thermal fluctuations. Our results reveal a universal optimal coding structure in thermalizing dynamics, while introducing new optimal codes that achieve fundamental limits of approximate quantum error correction. 
\end{abstract}

\maketitle

Quantum thermalization---the process by which an isolated quantum system effectively reaches quantum equilibrium by acting as its own bath \cite{polkovnikov_colloquium_2011, gogolin_equilibration_2016, nandkishore_many-body_2015, abanin_many-body_2019, alhambra_quantum_2023}---has been characterized using a variety of tools and frameworks, including quantum chaos and random matrix theory \cite{dalessio_quantum_2016, borgonovi_quantum_2016, linden_quantum_2009, srednicki_chaos_1994}, information scrambling \cite{maldacena_bound_2016}, and free probability \cite{pappalardi_eigenstate_2022}. While these approaches provide complementary perspectives to describe the dynamics of interacting many-body systems, a quantitative and operational understanding of thermalization remains incomplete.

Rapid advances in quantum information science provide us with new avenues for studying nonequilibrium many-body dynamics. State and unitary designs serve as the quantitative language for novel notions of quantum ergodicity, from Hilbert-space ergodicity and deep thermalization of state ensembles to pseudorandomness of the dynamics~\cite{dankert_exact_2009,brandao_local_2016,roberts_chaos_2017,mok_nature_2026,pilatowsky-cameo_complete_2023,pilatowsky-cameo_hilbert-space_2024,ho_exact_2022,ippoliti_dynamical_2023,cotler_emergent_2023,mark_maximum_2024}. In parallel, quantum error correction (QEC) provides a rigorous theory of the protection and recovery of quantum information in the presence of noise~\cite{leung_approximate_1997,kretschmann_information-disturbance_2006,Beny2010,crepeau2005approximate,bergamaschi_approaching_2022}. Despite these developments, the role of QEC as an operational framework for probing nonequilibrium many-body dynamics remains comparatively underexplored.

\begin{figure}[t!]
\includegraphics[width=\columnwidth]{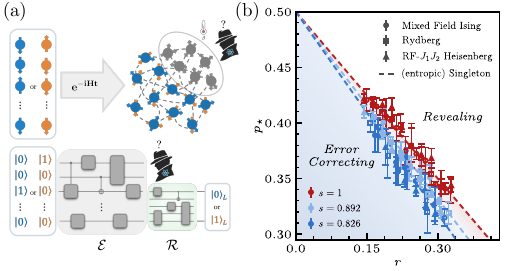}
  \caption{\textit{Thermalization as error correction.} (a) Under chaotic dynamics, initially distinct many-body states become locally indistinguishable at late times, forming codewords encoded by $U(t)=e^{-iHt}$. Erasure of a subsystem then acts as noise, while the encoded information can be recovered from the complementary subsystem. (b) Universal rate-distance, $r$-$p_\star$, tradeoff for thermalization codes generated by various ergodic Hamiltonians (parameters in the SM~\cite{SupplementalMaterial}). Thresholds from all Hamiltonians collapse onto a universal curve determined by entropy density $s$, which we identify as the (entropic) quantum Singleton bound -- showing that thermalization forms an asymptotically optimal approximate quantum error correcting code.}
    \label{fig:01}
\end{figure}

The connection between QEC and quantum thermalization initially seems counterintuitive: thermalization arises from generic, unstructured dynamics, while QEC relies on engineered, structured codes. Yet, both share an underlying information-theoretic principle. QEC hides logical information from local noise by making code states locally indistinguishable; thermalization does the same as subsystems equilibrate \cite{hayden_black_2007, ziman_quantum_2001, scarani_thermalizing_2002}.

In this work, we analyze quantum thermalization through the lens of QEC. We begin with a simple but rigorous observation: the late-time states of a thermalizing quantum system form an approximate quantum error-correcting code. Motivated by this connection, we numerically investigate when this correctability emerges in thermalizing many-body dynamics. We find that error correction is robust to differences in conserved quantities between codewords, failing only once these exceed thermal fluctuations. Within this regime, we uncover a universal relation between the encoding rate and distance of the emergent quantum code and the thermal entropy density. In the infinite-temperature limit, this relation reduces to the quantum Singleton bound \cite{knill_theory_1997}, achieved also by Haar-random codes \cite{ma_haar_2025}. At finite temperature, our identified universal relation deviates from the quantum Singleton bound. To obtain an analytic benchmark for this regime, we introduce \textit{Scrooge codes} and rigorously prove that these codes saturate the \textit{entropic quantum Singleton bound} \cite{grassl_entropic_2022}, establishing the optimal rate–distance tradeoff at fixed entropy density. The agreement between this rigorous bound and our numerical results shows that thermalization realizes the same optimal tradeoff. Together, these results identify thermalization as a natural mechanism for optimal approximate quantum error correction.

Our approach takes a distinct perspective from previous works that relate QEC to thermal physics. Prior works have taken the eigenstates of Hamiltonians satisfying the eigenstate thermalization hypothesis (ETH)~\cite{rigol_thermalization_2008} as codewords: ETH implies local indistinguishability within an energy window, hence approximate correctability of the corresponding code~\cite{ bao_eigenstate_2019,zheng_near-optimal_2024}. Other works have bounded correctability via covariance under continuous symmetries~\cite{faist_continuous_2020,liu_quantum_2023}, or via out-of-time-order correlators and the chaos bound~\cite{Qasim2025}. In our setting, however, the encoder is the physical time-evolution operator generated by an ergodic Hamiltonian, so the code states are dynamically produced from simple initial states, rather than selected from mid-spectrum eigenstates, which are not efficiently preparable. This allows us to ask whether approximate correctability can arise from thermalization itself, rather than engineered codewords.

\textit{Quantum thermalization as approximate quantum error correction.---} 
We consider a system of $N$ qubits, which we denote by $S$, interacting via a local Hamiltonian $H$, which generates unitary evolution $U(t)=e^{-iHt}$. Quantum thermalization states that simple initial states $\ket{\psi}$ (e.g., product states) evolve into states $\ket{\psi(t)}=U(t)\ket{\psi}$ which are locally indistinguishable from the Gibbs state
\begin{equation}
    g_\beta=\frac{e^{-\beta H}}{\tr(e^{-\beta H})},
\end{equation}
where $\beta$ is the effective inverse temperature, obtained by matching the initial state energy to the thermal energy, $\tr(Hg_\beta)=\expval{H}{\psi}$. Local indistinguishability means that, for any local region $E$ with complement $B$, for sufficiently late times,
\begin{equation}
\label{eq:thermalization}
    \norm{\tr_B(\dyad{\psi(t)})-g_{\beta,E}}_1\leq \epsilon,
\end{equation}
where $\norm{\, \cdot \,}_1$ is the trace norm and $\epsilon$ bounds a small, finite-size fluctuation. We say a state is \emph{$\epsilon$-thermalized} on $E$ if it satisfies Eq.~\eqref{eq:thermalization}.

A quantum error-correcting code encodes a logical Hilbert space into a physical codespace $\mathcal C$ via an encoding channel $\mathcal E$. Information is protected against a noise channel $\mathcal N$ if there exists a recovery channel $\mathcal R$ such that $\mathcal R\circ\mathcal N\circ\mathcal E$ acts as the identity on the logical space \cite{knill_theory_1997}.
 For approximate QEC, perfect recovery is relaxed to a small diamond-norm error   \cite{Beny2010}
 \begin{equation}
 \label{eq:approxerrorcorr}
  \exists~\mathcal{R\colon} \ \   \|\mathcal R \circ \mathcal N \circ \mathcal E - \mathrm{id}\|_\diamond \leq \delta.
 \end{equation} Here, $\|\cdot\|_\diamond$ characterizes the maximum distinguishability between two channels and implies that any state can be corrected to at most $\delta$ error. We take the noise channel to be erasure of a region $E$: $\mathcal{N}_B[\psi]\coloneqq \tr_E(\psi)$, with $\psi \coloneqq \dyad{\psi}$. We focus on correcting erasure errors because they reveal the most information to the environment. We consider time evolution itself as our encoding, $\mathcal{E}_t[\psi]=U(t)\psi U(t)^\dagger$, applied to $K$  orthogonal initial product states $\{\ket{\psi_k}\}_{k=1}^K$. These states evolve into the codespace $\mathcal C_t = \mathrm{span}\{U(t)\ket{\psi_k}\}_{k=1}^K$, with encoding rate $r = \log_2(K)/N$. Our first observation is that quantum thermalization implies approximate quantum error correction.
\begin{observation}
\label{obs:1} Suppose that every state in $\mathcal C_t$ is $\epsilon$-thermalized on $E$. Then $\mathcal C_t$ is approximately correctable against erasure of $E$ with recovery error
$\delta=2\sqrt{\epsilon \min\{K,2^{|E|}\}}$.
\end{observation}

A complete proof is provided in
the End Matter. Observation~\ref{obs:1} establishes a rigorous connection between quantum thermalization and QEC, but its assumptions also identify key limitations. First, the result requires all initial product states and their superpositions to thermalize locally to the same Gibbs state. However, codewords with different energies generally equilibrate to different Gibbs states and can therefore be partially distinguished by local energy-density measurements. This suggests a quantitative relation between correctability and the energy separation of the initial states, which we explore shortly. Second, the theorem assumes local indistinguishability on the erased region $E$. Since rigorous statements of quantum thermalization generally only apply for sufficiently small subsystems, it is important to determine how large $E$ can become before it ceases to appear thermal, leading to a breakdown of error correction. To explore these limitations, we turn to numerical analysis.

Directly evaluating the diamond-norm recovery error in Eq.~\eqref{eq:approxerrorcorr} is generally intractable for the system sizes needed for our analysis. We therefore diagnose correctability through the \emph{decoupling} of the erased region from an external reference. Consider the state $| \Psi(t) \rangle=\frac{1}{\sqrt K}\sum_{k=1}^K
\ket{\psi_k(t)}_{S} \ket{k}_R$,
where the auxiliary Hilbert space $R$ of dimension $K$ purifies the maximally mixed logical input. The decoupling principle states that exact recovery from $B$ is possible if and only if, after encoding, the quantum mutual information between the erased region and the reference is zero, $I(E:R) = 0$ \cite{schumacher_quantum_1996, hayden_decoupling_2008}. In the Supplementary Material (SM)~\cite{SupplementalMaterial}, building on Refs.~\cite{schumacher_approximate_2001,barnum_quantum_2000,ahlswede_quantum_2013}, we prove that small quantum mutual information implies correctability in diamond norm on a single subcode of the same asymptotic rate, simultaneously for all possible erasures of fixed size. This lets us certify approximate correctability in diamond norm directly from $I(E:R)$, without explicitly constructing a recovery channel, providing a numerically accessible proxy for correctability which we use throughout.

\begin{figure}[t!]
\includegraphics[width=\columnwidth]{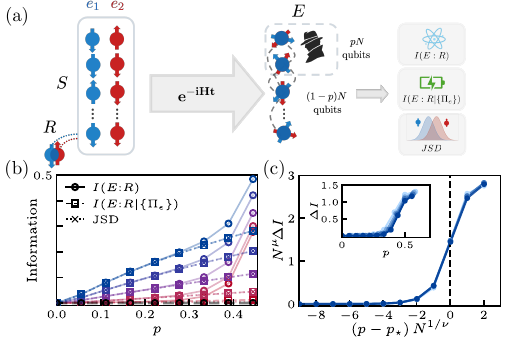}
  \caption{\textit{Quantifying information leakage due to energy.} (a) A logical qubit is encoded in two energy-separated codewords maximally entangled with a reference $R$, then evolved. After erasing a fraction $p$ to $E$, we compare total mutual information, $I(E:R)$ (quantum), energy measurement-restricted mutual information, $I(E:R|\{\Pi_e\})$ (classical-quantum), and the corresponding Jensen--Shannon divergence, JSD, (classical) between the local-energy outcomes. (b) For $N=18$ mixed-field Ising evolution, we plot $I(E:R)$, $I(E:R|\{\Pi_e\})$, and the local-energy JSD versus erased fraction $p$ for $p< 1/2$, which agree up to a finite-size threshold before deviating. The codeword energy-density separation is centered at zero and increases from red to blue: $\Delta e=0.09,0.27,0.45,0.63,0.78$. The black dashed line shows $I(E:R)$ for 
$\Delta e=0$. (c) Finite-size scaling of $\Delta I \equiv I(E:R)-\mathrm{JSD}$ extracts the threshold $p_\star$ at which agreement holds in the thermodynamic limit. For $N=12,14,16,18$ and $\Delta e=0.9045$, $\Delta I$ collapses as $N^\mu\Delta I$ versus $(p-p_\star)N^{1/\nu}$, with numerically extracted $p_\star\simeq0.5$, $\mu\simeq0.3$, $\nu\simeq1$. Inset: unscaled $\Delta I$ versus $p$.}
    \label{fig:02}
\end{figure}

\textit{Codewords with different effective temperature.---} 
To begin our numerical exploration, we relax the assumption that all codewords equilibrate to the same thermal state by considering initial states with different energies, focusing on the simplest case $K=2$. We take two orthogonal initial product states $\ket{\psi_1}$ and $\ket{\psi_2}$ with energy separation $\Delta E_N$, which we take to be centered about zero, i.e., $\langle H\rangle_{\psi_{1,2}} = \pm \Delta E_N/2$. Concretely, we choose $H$ to be the nonintegrable mixed-field Ising Hamiltonian. Here, zero energy corresponds to the infinite temperature, $\beta = 0$, energy. We evolve the states under $U(t) = e^{-iHt}$ to late times $t \sim N^2/J$, where $J$ is the characteristic energy scale of the Hamiltonian. See the SM~\cite{SupplementalMaterial} for details on the Hamiltonian parameters, initial state construction, and numerical protocol. In Fig.~\ref{fig:02}(b), we show $I(E:R)$ as a function of the fraction of erased qubits $p =  |E|/N$ and the energy density gap $\Delta e= \Delta E_N/N$. When $\Delta e = 0$,  $I(E:R) = 0$ and error correction is achieved, as predicted by Observation~\ref{obs:1}. For $\Delta e\neq 0$, by contrast, $I(E:R)$ grows with both $\Delta e$ and $p$, showing that the environment can distinguish codewords of different energies. This motivates a more refined analysis of what information is actually revealed to the environment.

We find numerically that the information revealed to the environment is contained entirely in measurements of the local energy density. To make this precise, we define the measurement-restricted mutual information $I(E:R|\{\Pi_e\})$ as the information obtained by an environment restricted to measuring the subsystem energy, with ${\Pi_e}$ the eigenspace projectors of the terms in $H$ fully supported within $E$. We show this simply reduces to the Holevo information~\cite{SupplementalMaterial}. We also consider a fully classical version obtained by further dephasing the reference $R$ in the codeword-label basis. For $K=2$, this quantity is the Jensen--Shannon divergence (JSD) between the two probability distributions of local energy outcomes on $E$ for states with differing total  energy. Figs.~\ref{fig:02}(b) and \ref{fig:02}(c) show that the three quantities numerically coincide up to $p = 1/2$ in the thermodynamic limit for constant $\Delta e$, using a finite-size scaling analysis. The agreement of $I(E:R|\{\Pi_e\})$ and $I(E:R)$ for $p\leq 1/2$ (beyond $p>1/2$, error correction is impossible due to the no-cloning theorem) reveals the mechanism by which information is revealed to the environment --- the only information available is the classical energy difference.

Since the JSD agrees with the mutual information for $p \leq 1/2$, we use it to determine the largest possible energy separation $\Delta E_N$ such that error correction is preserved in the thermodynamic limit. To quantify the information contained in the subsystem energy, we construct a coarse-grained high-temperature model in which the codeword-dependent energy distributions are Gaussian with variance $O(N)$ and mean separation proportional to $\Delta E_N$. Within this model, the JSD can be evaluated analytically, yielding the asymptotic scaling: for $\Delta E_N \sim N^\kappa$, the JSD vanishes for $\kappa<1/2$, remains finite at $\kappa = 1/2$, and approaches one bit for $\kappa >1/2$ \cite{SupplementalMaterial}. Intuitively, thermal energy fluctuations scale as $\sqrt{N}$, so an $o(\sqrt{N})$ separation between codeword energies cannot be resolved against these fluctuations, and error correction persists. This helps explain why ETH codes \cite{brandao_quantum_2019} succeed: restricting to eigenstates within a narrow energy window gives exponentially small energy differences, far below the $\sqrt{N}$ fluctuation scale, placing their codewords deep within the error-correcting regime. In the SM~\cite{SupplementalMaterial}, we extend the analysis beyond energy to a finite set of commuting extensive conserved charges, replacing the subsystem-energy distributions by the corresponding joint subsystem-charge distributions.

\begin{figure}[t!]
\includegraphics[width=\columnwidth]{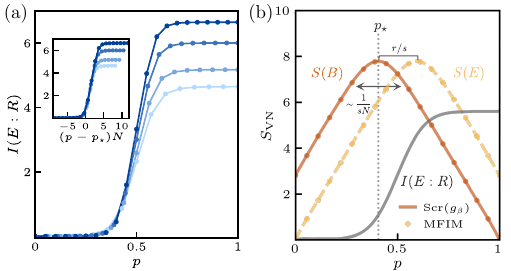}
  \caption{(a) Infinite temperature information transition. For $N = 14, 16, 18, 20$ (from the lightest to the darkest curve), we construct codespaces from orthonormal product states such that the energy is nearly independent of the encoded logical state and evolve them under the ergodic mixed-field Ising Hamiltonian at $\beta =0$ and $r = 0.1667$. We plot $I(E:R)$ as a function of $p$. Inset: finite-size scaling collapse as a function of $(p-p_\star)N^{1/\nu}$, with numerically extracted values $p_\star \simeq 0.417$, $\nu \simeq 1$. (b) Generalized Page-curve picture of the information transition. For a rate-$r$ $g_\beta$-Scrooge code, we plot $S(E)$, $S(B)$, and the resulting $I(E:R)$ versus erased fraction $p$; dots schematically represent mixed-field Ising entropies. The two entropy curves have slopes set by the thermal entropy density $s$, and the finite code dimension $\log_2 K=rN$ provides the offset that shifts the transition to the entropic Singleton bound $p_\star=(s-r)/2s$. At $s=1$, Scrooge reduces to Haar and the bound recovers $p_\star = (1-r)/2$.}
    \label{fig:03}
\end{figure}

\textit{Universal error-correcting transition.---} 
Having established the regime in which correctability survives, we now ask how well thermalization performs as an error-correcting code. Specifically, for codespaces in which the energy is approximately independent of the encoded logical state, what are the largest encoding rate $r$ and erased fraction $p$ for which the thermalizing dynamics remain error correcting? In the SM~\cite{SupplementalMaterial}, we explicitly construct orthogonal initial product states contained within a narrow energy window, such that any coherent superposition of time-evolved codewords has nearly the same energy expectation value.

 We find that, for constant rate $r$ and below a critical fraction $p< p_\star$, the mutual information $I(E:R)$ vanishes, and thus the code $\mathcal{C}_t$ is error correcting. As $p$ increases past $p_{\star}$,  $I(E:R)$ grows until it saturates at its maximum value $2 Nr$, as depicted in Fig.~\ref{fig:03}(a) for fixed $r$ and $\beta=0$, where finite-size analysis yields the information transition point $p_\star$. Repeating this analysis across a range of encoding rates $r$ and for three different ergodic Hamiltonians \cite{SupplementalMaterial}, we empirically find the erasure threshold $p_\star$ is described by
 \begin{equation}
 \label{eq:univ-behavior}
     p_\star(r)=\frac{s-r}{2s},
 \end{equation}
 where $s = S(g_\beta)/N$ is the von Neumann entropy density of the Gibbs state, as shown in  Fig.~\ref{fig:01}(b). Despite differences in microscopic dynamics, the extracted thresholds for all three ergodic Hamiltonians collapse onto the same curve, determined solely by $r$ and $s$.

At infinite temperature, where $s = 1$, Eq.~\eqref{eq:univ-behavior} is in fact the theoretical upper bound for error correction, as dictated by the quantum Singleton bound $p_\star=(1-r)/2$ \cite{knill_theory_1997}. Haar-random codes have been shown to saturate this bound \cite{ma_haar_2025}. Importantly, however, Haar-random codes are not physically feasible, due to the exponential circuit complexity of Haar-random states. Here, we show that the same level of error correction is achieved by states that a physical system reaches on its own, under local Hamiltonian evolution.

Beyond establishing optimality, we can use Haar-random codes to understand the threshold directly. Since Haar-random entropies are known exactly, $I(E:R)$ can be understood analytically through $S(E)$ and $S(B)$. For a Haar-random orthonormal code subspace $\{| \bar{\psi}_k \rangle\}_{k=1}^K$ with $K\ll 2^N$, the mixed reduced state
$
\rho_E^{\mathrm{mix}}=\tr_B\!\left(\frac{1}{K}\sum_k | \bar{\psi}_k \rangle \hspace{-.08cm} \langle \bar{\psi}_k |\right)
$
is well approximated by the reduced state of a single Haar-random state on an enlarged Hilbert space of dimension $2^{N+\log_2 K}$ \cite{SupplementalMaterial}. The code subspace therefore adds $\log_2 K$ qubits to the complement, shifting the Page curve \cite{page_average_1993} and hence the transition point. At infinite temperature, $s=1$, these curves give $S(E) \simeq N \min \{p, 1-p+r\}$ and $S(B) \simeq N \min \{1-p, p+r\}$, as illustrated in Fig.~\ref{fig:03}(b). Since the global state is pure, $I(E:R) = S(E) + rN - S(B)$. For $p < (1-r)/2$, $S(E) \simeq pN$ and $S(B) \simeq (p+r)N$, so these contributions cancel and $I(E:R) \simeq 0$. The transition occurs when $S(B)$ reaches its Page maximum, $(1-p)N = (p+r)N$, giving $p_\star = (1-r)/2$. Above this point, $I(E:R)$ grows linearly until $S(E)$ reaches its own Page maximum at $p = (1+r)/2$, after which it saturates at $2 \log_2 K$. The $O(1)$ Page correction rounds this transition over a width of $O(N^{-1})$, consistent with the finite-size scaling in Fig.~3(a). At infinite temperature, $s=1$, the thermalization entropies closely follow these shifted Haar Page curves, explaining the observed information transition.

At finite temperature, the decrease in thermal entropy density $s<1$ reduces the effective Hilbert space and fundamentally alters the landscape of achievable codes. Nevertheless, repeating the numerical analysis for finite-temperature codewords, we find that all ergodic Hamiltonians we consider again collapse onto a universal curve determined only by $s$, Fig.~\ref{fig:01}(b). To understand the origin of this bound and whether it is optimal, we introduce Scrooge codes, an analytically tractable construction based on the Scrooge ensemble.

 {\it Scrooge codes.---}
 The Scrooge ensemble~\cite{jozsa_lower_1994,goldstein_distribution_2006,reimann_typicality_2008} $\mathrm{Scr}(\rho)$ is the unique ensemble over pure states $|\phi\rangle\in\mathrm{supp}(\rho)$ with first moment $\rho$ that minimizes accessible information. Writing $d_\rho=\rank(\rho)$ and $\mathrm d\psi$ for normalized Haar measure on the unit sphere of $\mathrm{supp}(\rho)$,
\begin{equation}
\mathrm{Scr}(\rho) = \left\{d_\rho\bra{\psi}\rho \ket{\psi}\;\mathrm{d}\psi, \ket{\phi}=\frac{\sqrt{\rho}\ket{\psi}}{\|\sqrt{\rho}\ket{\psi}\|} \right\}
\end{equation}
so that $\mathbb E_{|\phi\rangle\sim\mathrm{Scr}(\rho)}
\left[\ketbra{\phi}\right]=\rho$. Typical samples are locally indistinguishable from $\rho$~\cite{reimann_typicality_2008, mcginley_scrooge_2025}, so $\mathrm{Scr}(g_\beta)$ provides a natural random-matrix model for finite-temperature late-time states~\cite{mark_maximum_2024, mok_nature_2026}.

Using this ensemble, we define a family of \emph{Scrooge codes}. A rate-$r$ $\rho$-Scrooge code is obtained by drawing $K=2^{rN}$ independent samples $|\phi_k\rangle\sim\mathrm{Scr}(\rho)$ and replacing them by the closest orthonormal set $\{|\widehat\phi_k\rangle\}_{k=1}^K$, defined as the minimizer of $\sum_{k=1}^K \big\| |\widehat\phi_k\rangle-|\phi_k\rangle \big\|^2$. This produces a valid coding isometry. We show that a $g_\beta$-Scrooge code saturates the entropic quantum Singleton bound, provided a simple, physically motivated condition on the Gibbs state holds. Specifically, we say that $\rho$ has an \textit{entropy-concentrated spectrum} if all but vanishing spectral weight lies on eigenvectors whose eigenvalues $\lambda_i$ satisfy $- \log \lambda_i = S(\rho) + o(N)$. We assume this holds for the reduced states $\tr_B(g_\beta)$ and $\tr_E(g_\beta)$ appearing in the erasure, and that both have entropy density $s(\beta)>0$~\cite{SupplementalMaterial}.
\begin{theorem}\label{thm:scrooge_achievability_mt}
Let $g_{\beta,N}$ satisfy the assumptions above, and let $\mathcal C_N$ be a $g_{\beta,N}$-Scrooge code of rate $r<s_\infty(\beta)$, where
$S_\infty(g_{\beta,N})/N\to s_\infty(\beta)$. For every
$p<p_\star(r)$, where $p_\star(r)$ is given by Eq.~\eqref{eq:univ-behavior}, with probability
tending to one, $\mathcal C_N$ simultaneously decouples the logical
reference from every erasure of at most $pN$ qubits. Consequently, there exists a single subcode of
$\mathcal C_N$ with the same asymptotic rate
that is recoverable from every such erasure with vanishing
diamond-norm error.
\end{theorem}

To establish optimality, we use the entropic quantum Singleton bound, which sets the optimal rate--distance limit for codes subject to a bounded-entropy constraint~\cite{grassl_entropic_2022}. We assume that two disjoint erased regions of size $pN$ leave a complement of size $(1-2p)N$ with entropy density $s(\beta)$ and an entropy-concentrated spectrum, allowing the entropic Singleton bound to yield $p \le p_\star(r)$~\cite{SupplementalMaterial}. Since Theorem~\ref{thm:scrooge_achievability_mt} achieves every $p<p_\star(r)$, $g_\beta$-Scrooge codes are asymptotically optimal.

Repeating the entropy analysis for $g_\beta$-Scrooge codes gives a simple analytic approach to understand the finite-temperature transition, as follows. When $K \operatorname{tr} (g_\beta^2)\ll 1$, the reduced state of the mixture of $K$ $g_\beta$-Scrooge codewords is well approximated by that of a single enlarged Scrooge-random state with parent $g_\beta\otimes I_K/K$ \cite{SupplementalMaterial}. The code therefore adds $\log_2 K= rN$ qubits to the complement, shifting the subsystem entropy exactly as in the infinite-temperature case~\cite{SupplementalMaterial, nakagawa_universality_2018,fujita_page_2018}. With slopes set by $s < 1$, this yields the entropy curves shown in Fig.~\ref{fig:03}(b), which agree with the finite-temperature thermalization entropies. These curves reproduce the mutual information transition: the turnover of $S(B)$ occurs at $p_\star=(s-r)/(2s)$, giving the observed threshold. The finite-temperature correction further gives a crossover width $O(1/(sN))$, consistent with the observed finite-size scaling of the transition.

\textit{Discussion and Outlook.---} There are several natural directions for further investigation. First, it remains to understand the timescales governing the onset of quantum error correction. Since local thermalization makes erased subsystems insensitive to the encoded state, the time at which $I(E:R)$ becomes small should be closely related to the thermalization time. However, finite-rate recoverability may also depend on the buildup of nonlocal correlations across the complement. Determining whether correctability emerges with local thermalization, with scrambling, or on a distinct timescale would sharpen the connection between dynamics and quantum error correction \cite{sahu_phase_2024}. Second, additional conservation laws, including non-Abelian symmetries and hydrodynamic slow modes \cite{doyon_generalized_2025}, should be incorporated into this error-correction framework. While commuting conserved quantities introduce only classical leakage through their joint distributions, more general symmetry structures may impose richer constraints on error correction. Finally, extending this framework to integrable, many-body localized \cite{nandkishore_many-body_2015, abanin_many-body_2019}, and quantum-scarred systems \cite{serbyn_quantum_2021} may provide an operational diagnostic of nonergodicity through the failure, weakening, or delay of emergent error correction.

\textit{Acknowledgments.---} We thank Sarang Gopalakrishnan, Yu-Jie Liu, Daniel Mark, Daniel Ranard, and Xinyu Tan for helpful discussions. AV acknowledges support from the National Science Foundation Graduate Research Fellowship Program under Grant No. DGE-2140743. RRA was supported by a grant from the Simons Foundation (MP-SIP-00001553, AWH) and by the U.S. Department of Energy, Office of Science, National Quantum Information Science Research Centers, Co-design Center for Quantum Advantage (C2QA) under contract number DE-SC0012704. This work was also supported by the Center for Ultracold Atoms (an NSF Physics Frontiers Center; PHY-2317134), the Heising-Simons Foundation (Grant No. 2024-4851), and the Alfred P. Sloan Foundation through a Sloan Research Fellowship.

\bibliography{references}
\nocite{Ng_2010, vershynin_high-dimensional_2018, ekert_error_1996, lowdin_nonorthogonality_1970, tropp_user-friendly_2012, winter_coding_1999, hoeffding_probability_1963, meckes_spectral_2013, anshu_concentration_2016, Kliesch2014, bjelakovic_shannon-mcmillan_2002, sugiura_canonical_2013, holevo_bounds_1973, Watrous2018TQI, mele_introduction_2024, de_maesschalck_mahalanobis_2000, shannon1948mathematical}

\appendix
\section{End Matter}
\label{app:observation-one}

Here we give the complete proof of Observation~\ref{obs:1}. 

\begin{proof} Due to thermalization, the identity of a codeword $\ket{\psi_k(t)}$ cannot be revealed through measurements of $E$; therefore, it must be recoverable solely from the reduced state in $B$. This intuition is formalized by the {\it information-disturbance tradeoff} for quantum channels~\cite{kretschmann_information-disturbance_2006} as follows. Define the {\it complementary channel} $\mathcal{N}_E[\psi]\coloneqq \tr_B[\psi]$ by tracing out $B$ instead of $E$. Because every state in the codespace thermalizes, the channel $\mathcal{N}_E\circ \mathcal{E}_t$ is close in diamond distance to the constant channel $\mathcal{G}[\rho]:= \tr(\rho)g_{\beta,E}$:
\begin{align*} \norm{\mathcal{N}_E\circ \mathcal{E}_t-\mathcal{G}}_{\diamond}  &\leq \min\{K,2^{|E|}\}\,{\sup_{\rho}} \norm{\mathcal{N}_E\circ \mathcal{E}_t[\rho]-\,g_{\beta,E}}_1 \\&\leq \epsilon \min\{K,2^{|E|}\}, \end{align*}
where the first inequality follows from standard relations between the diamond and trace norms~\cite{Watrous2018TQI}, and the supremum is taken over all density matrices $\rho$ on $\mathcal{C}_0$, which also thermalize by triangle inequality. Under these conditions, Theorem 3 of Ref.~\cite{kretschmann_information-disturbance_2006} guarantees that there is a recovery channel $\mathcal R$ for the channel $\mathcal{N}_B\circ \mathcal{E}_t$ such that \begin{align*}
   \left\|
\mathcal R\circ\mathcal N_B\circ\mathcal E_t-\mathrm{id}
\right\|_\diamond
&\le 2
\norm{\mathcal{N}_E\circ \mathcal{E}_t-\mathcal{G}}^{1/2}_{\diamond}\\
&\le 2\sqrt{\epsilon \min\{K,2^{|E|}\}}.\qedhere
\end{align*} 
\end{proof}

\end{document}


\title{Supplementary Material for ``Quantum thermalization achieves optimal approximate quantum error correction''}
\author{Aditi Venkatesh}
\thanks{Electronic address: avenkatesh@g.harvard.edu}
\affiliation{Department of Physics, Harvard University, Cambridge, MA, 02138}
\affiliation{Center for Theoretical Physics --- a Leinweber Institute, Massachusetts Institute of Technology, Cambridge, MA, 02139}

\author{Richard R. Allen}
\affiliation{Center for Theoretical Physics --- a Leinweber Institute, Massachusetts Institute of Technology, Cambridge, MA, 02139}

\author{Saúl Pilatowsky-Cameo}
\affiliation{Center for Theoretical Physics --- a Leinweber Institute, Massachusetts Institute of Technology, Cambridge, MA, 02139}

\author{Bingtian Ye}
\affiliation{Center for Theoretical Physics --- a Leinweber Institute, Massachusetts Institute of Technology, Cambridge, MA, 02139}

\author{Soonwon Choi}
\affiliation{Center for Theoretical Physics --- a Leinweber Institute, Massachusetts Institute of Technology, Cambridge, MA, 02139}
\maketitle

In this Supplemental Material, we provide the technical details supporting the claims and results presented in the main text. In Sec.~\ref{sec:approximate_qec}, we review average- and worst-case approximate quantum error correction, establish their asymptotic equivalence at positive rate, formulate the decoupling criterion used throughout our analysis, and develop the entropic quantum Singleton converse needed for codes with bounded local entropy. In Sec.~\ref{sec:scrooge_achieve}, we introduce Scrooge random codes and derive achievability bounds from their second moments and a finite-size decoupling inequality. In Sec.~\ref{sec:tightness}, we prove Theorem~1 of the main text, establishing saturation of the entropic quantum Singleton bound by bounded-entropy and thermal Scrooge codes. In Sec.~\ref{sec:thermalization_codes}, we describe the numerical construction of thermalization codes, the chaotic Hamiltonians used in our simulations, the finite-size scaling analysis, and comparisons with Haar-random and Scrooge-random codes. Finally, in Sec.~\ref{sec:conservation_constraints}, we analyze how conserved quantities constrain error correction, characterize the resulting information leakage through local measurements, and extend the analysis to multiple conserved quantities and basis-dependent dephasing.

\tableofcontents
\clearpage
\section{Approximate Quantum Error Correction}
\label{sec:approximate_qec}
We begin by formally defining approximate quantum error correction (AQEC). An $[[n,k]]$ approximate quantum error correcting code is specified by an isometric \emph{encoding} channel $\mathcal E: \cL(\C^{K}) \to \cL(\C^{D})$ together with a \emph{recovery}, or \emph{decoding}, channel $\mathcal R : \mathcal L(\mathbb C^D) \to \mathcal L(\mathbb C^K)$ where $K = 2^k$, and $D = 2^n$. The \emph{codespace} $\cC$ of the code is defined as the image of the \emph{encoding isometry} $V$ such that $\cE(\rho) = V \rho V^\dagger$. In what follows, we will consider error correcting codes for a single, fixed noise channel $\cN: \cL(\C^D) \to \cL(\C^D)$, e.g., the erasure channel on some subset of qubits.

We say that $(\cE,\cR)$ is an \emph{exact} QECC for $\cN$ if $\mathcal R \circ \mathcal N \circ \mathcal E = \mathrm{id}$, for $\mathrm{id}$ the identity channel on $\cL(\C^K)$. In other words, any state in the codespace $\cC$ (which may additionally be entangled with an ancillary system) can be exactly recovered even after the noise channel $\cN$ is applied. For AQEC, we relax this notion, to say that states in the codespace need only be approximately recoverable after $\cN$ is applied. There are two notions of approximate recoverability: \emph{average-case} recoverability and \emph{worst-case} recoverability. The former has historically been used to define AQEC~\cite{schumacher_approximate_2001} and has the benefit of being more readily computable, while the latter has since become more standard~\cite{crepeau2005approximate, Beny2010, Ng_2010}. For exact QEC these notions are evidently equivalent, since the average error is zero if and only if the worst-case error is zero, but this is not obviously the case for AQEC.

In the following section, we precisely define average-case and worst-case AQEC and establish their asymptotic equivalence at positive rate. More specifically, we show that a positive-rate sequence of codes that is uniformly average-case correctable against at most exponentially many noise channels contains a single subcode of the same asymptotic rate that is uniformly worst-case correctable against all of them. The reduction incurs only a subextensive loss in the number of logical qubits, while the worst-case error still vanishes asymptotically. Thus, for the purposes of asymptotic code construction, it suffices to establish average-case AQEC.

\subsection{Worst to Average-Case Reduction for Approximate Quantum Error Correction}

We begin by defining average- and worst-case AQEC. We will characterize both notions using the \emph{entanglement fidelity}. Entanglement fidelity is a property of a quantum state $\rho$ and channel $\Lambda$ acting on the Hilbert space of $\rho$, defined as
\begin{align}
    F_e(\rho, \Lambda) = \bra{\phi} \Lambda \otimes \mathrm{id}(\ketbra{\phi})\ket{\phi},
\end{align}
for $\ket{\phi}$ a purification of $\rho$ using an ancillary Hilbert space, on which the identity channel acts above (the definition is independent of purification). We first define average-case AQEC:

\begin{defn}[Average-Case AQEC]\label{def: avg_case} $(\cE, \cR)$ is an $\eps$-approximate average-case quantum error correcting code for noise channel $\cN$ if
\begin{align}
    \left(1 - F_e\left(\id/K, \mathcal R\circ\mathcal N\circ\mathcal E \right) \right)^{1/2} \le \eps.
\end{align}
\end{defn}
\noindent
Note that the canonical purification of the maximally mixed state $\id/K$ is simply the maximally entangled state across two copies of $\C^K$; the average-case error correcting performance of a code therefore characterizes its ability to share and preserve entanglement. A high entanglement fidelity for the maximally mixed state ensures high average-case code performance because the entanglement fidelity lower bounds the ensemble average fidelity for any pure state decomposition of $\rho$~\cite{Watrous2018TQI}. In other words, if $(\cE,\cR)$ is an $\eps$-approximate code, then we have
\begin{align}
    \frac{1}{K} \sum_{k = 0}^{K-1} \bra{k}  \mathcal R\circ\mathcal N\circ\mathcal E(\ketbra{k}) \ket{k} \ge 1 - \eps^2,
\end{align}
i.e., the average fidelity of a randomly transmitted basis codeword is at least $1 - \eps^2$. 

By contrast, we define worst-case AQEC in terms of the entanglement fidelity optimized over all states:
\begin{defn}[Worst-Case AQEC]\label{def: worst_case} $(\cE,\cR)$ is an $\eps$-approximate worst-case quantum error correcting code for noise channel $\cN$ if
\begin{align}
    \max_{\rho \in \cD(\C^K)} \left(1 - F_e\left(\rho, \mathcal R\circ\mathcal N\circ\mathcal E\right) \right)^{1/2} \le \eps.
\end{align}
\end{defn}
\noindent
We note that, by the Fuchs--van de Graaf inequalities relating trace distance and fidelity~\cite{Watrous2018TQI}, this definition of worst-case AQEC implies that $\Lambda = \mathcal R\circ\mathcal N\circ\mathcal E$ is close to the identity channel in diamond distance:
\begin{align}
     \| \Lambda - \mathrm{id} \|_\diamond &= \max_{\substack{\ket{\psi} \in \C^K \otimes \C^K \\ \langle\psi|\psi\rangle = 1}} \|  (\Lambda\otimes \mathrm{id})(\ketbra{\psi})-\ketbra{\psi} \|_1
    \\
    &\le 2\max_{\substack{\ket{\psi} \in \C^K \otimes \C^K \\ \langle\psi|\psi\rangle = 1}}
    \left(
    1-\bra{\psi}(\Lambda\otimes \mathrm{id})(\ketbra{\psi})\ket{\psi}
    \right)^{1/2} 
    \\
    &=
    2\max_{\rho\in\mathcal D(\mathbb C^K)}
    \left(1-F_e(\rho,\Lambda)\right)^{1/2} \\
    &\le 2\epsilon. 
\end{align}

We now show that, for asymptotic code families, uniform average-case correctability against exponentially many noise channels implies uniform worst-case correctability on a single subcode of the same asymptotic rate.
\begin{theorem}[Worst to Average-Case Reduction for AQEC]\label{thm: worst_to_avg}
Let $\mathcal E_N: L(\mathbb C^{K_N}) \rightarrow L(\mathcal H_N)$ be a sequence of $K_N$-dimensional encodings with asymptotic rate $r$, and let $\mathcal N_{N,a}$ for $a \in \mathcal A_N$ be families of at most exponentially many noise channels. Suppose that each $\mathcal N_{N,a}$ admits a recovery $\mathcal R_{N,a}$ such that, uniformly in $a$,
\begin{equation}
    F_e\left(\frac{\mathbf \id_{K_N}}{K_N},
\mathcal R_{N,a}\circ\mathcal N_{N,a}\circ\mathcal E_N\right)\longrightarrow 1
\end{equation}
Then, there exists a sequence of subcodes $\mathcal C'_N$ with encoding channels $\cE_N'$, of the same asymptotic rate $r$, that is uniformly worst-case correctable against every $\mathcal N_{N,a}$. In particular, appropriate recovery channels satisfy, 
\begin{equation}
    \max_{a\in\mathcal A_N}
\left\|
\mathcal R'_{N,a}\circ\mathcal N_{N,a}\circ\mathcal E'_N
-\operatorname{id}
\right\|_\diamond
\longrightarrow0.
\end{equation}
\end{theorem}
We prove the above theorem over the course of two lemmas. 
\begin{lemma}[Adapted from Lemma 21 in \cite{ahlswede_quantum_2013}]\label{lemm:haar_subcode_space} Let ${\Lambda_{N,a}}$, ${a\in\mathcal A_N}$ be a family of channels on $\mathbb C^{K_N}$, where $\log |\mathcal A_N|=O(N)$ and $N^{-1}\log_2 K_N\to r>0$. Suppose
\begin{equation}
\epsilon_N
:=
\max_{a\in\mathcal A_N}
\left(
1-F_e\left(\frac{\id_{K_N}}{K_N},\Lambda_{N,a}\right)
\right)^{1/2}
\longrightarrow 0.
\end{equation}
Then, for all sufficiently large $N$, there exists a single subspace $\mathcal C'_N\subset\mathbb C^{K_N}$ of dimension $K'_N=\left\lfloor\frac{K_N}{N^3}\right\rfloor$ such that 
\begin{equation}
\min_{a\in\mathcal A_N}
\min_{\substack{|\psi\rangle\in\mathcal{C}'_N \\ \langle\psi|\psi\rangle = 1}}
\langle\psi|\Lambda_{N,a}(\dyad{\psi})|\psi\rangle
\geq
1-\epsilon_N^2-\frac{1}{N}.
\end{equation}
\end{lemma}
\begin{proof}
    We adapt the strong-subspace argument of~\cite{ahlswede_quantum_2013} to an arbitrary exponentially sized family of channels, using a direct net argument. For each $a\in\mathcal A_N$, define the pure-state channel infidelity $g_{N,a}(\psi) := 1- \langle\psi| \Lambda_{N,a}(\dyad{\psi}) |\psi\rangle$. The relation between the Haar-averaged pure-state fidelity and entanglement fidelity gives, \begin{equation}
        \mathbb E_{\psi} \left[ \langle\psi| \Lambda_{N,a}(\dyad{\psi}) |\psi\rangle \right] = \frac{ K_N F_e\left( \id_{K_N}/K_N, \Lambda_{N,a} \right) +1 }{K_N+1}
    \end{equation}
    where the expectation is over Haar-random unit vectors in $\mathbb C^{K_N}$ \cite{mele_introduction_2024}. Therefore, 
    \begin{equation}
        \mathbb E_{\psi} \left[ g_{N,a}(\psi) \right] = \frac{K_N}{K_N+1} \left( 1- F_e\left( \frac{\id_{K_N}}{K_N}, \Lambda_{N,a} \right) \right) \leq \epsilon_N^2.
    \end{equation}
    We next show that $g_{N,a}$ is uniformly Lipschitz. For unit vectors $\ket{\psi}$ and $\ket{\phi}$, \begin{align}
        \left|g_{N,a}(\psi)-g_{N,a}(\phi) \right| &\leq \left|\tr\left[(\dyad{\psi} - \dyad{\phi}) \Lambda_{N,a}(\dyad{\psi})\right] \right| + \left| \operatorname{tr} \left[ \dyad{\phi} \Lambda_{N,a}(\dyad{\psi}-\dyad{\phi}) \right] \right|.
    \end{align}
    By Hölder's inequality, positivity of $\Lambda_{N,a}(\dyad{\psi})$, and trace-norm contractivity of quantum channels on Hermitian operators,
    \begin{align} \left| g_{N,a}(\psi)-g_{N,a}(\phi) \right| &\leq 2\|\dyad{\psi}-\dyad{\phi}\|_1\\
    &=4 \sqrt{1 - |\langle \psi|\phi\rangle|^2}\\
    &\leq 4\|\dyad{\psi}-\dyad{\phi}\|_2.
    \end{align}
    Thus, $g_{N,a}$ is $4$-Lipschitz, uniformly in $a$ and $N$.
    Let $W_N:\mathbb C^{K'_N}\longrightarrow\mathbb C^{K_N}$ be a Haar random isometry. For every fixed unit vector $\ket{x} \in \mathbb C^{K'_N}$, the vector $W_N \ket{x}$ is Haar-random on the unit sphere of $\mathbb C^{K_N}$. Concentration of measure on the unit sphere \cite{vershynin_high-dimensional_2018} therefore implies that, for some universal constant $c>0$,
    \begin{equation}
        \Pr_{W_N} \left\{ g_{N,a}(W_Nx) > \mathbb E_\psi[g_{N,a}(\psi)] + \frac{1}{2N} \right\} \leq 2\exp\left( -c\frac{K_N}{N^2} \right).
    \end{equation}
    Using $\mathbb E_\psi[g_{N,a}(\psi)]\leq\epsilon_N^2$, we obtain \begin{equation} \Pr_{W_N} \left\{ g_{N,a}(W_Nx) > \epsilon_N^2+\frac{1}{2N} \right\} \leq 2\exp\left( -c\frac{K_N}{N^2} \right). \end{equation}
    Let $\mathcal M_N$ be a $1/(8N)$-net of the unit sphere of $\mathbb C^{K'_N}$. By the standard volumetric bound, the unit sphere in $\mathbb R^d$ admits a $\delta$-net containing at most $(1+2/\delta)^d$ points~\cite{vershynin_high-dimensional_2018}. Since $\mathbb C^{K'_N}\cong\mathbb R^{2K'_N}$ as a real vector space, taking $\delta=1/(8N)$ allows us to choose $\mathcal M_N$ such that $|\mathcal M_N|\leq (1+16N)^{2K'_N}$. Taking a union bound over all $a\in\mathcal A_N$ and $x\in\mathcal M_N$ gives
    \begin{align} &\Pr_{W_N} \left\{ \exists\,a\in\mathcal A_N,\, x\in\mathcal M_N: g_{N,a}(W_Nx) > \epsilon_N^2+\frac{1}{2N} \right\} \leq 2|\mathcal A_N| (1+16N)^{2K'_N} \exp\left( -c\frac{K_N}{N^2} \right). \end{align}
    The logarithm of the right hand side is $O(N) + 2K'_N\log(1+16N) - c\frac{K_N}{N^2}$. Since $K'_N\leq K_N/N^3$, this is bounded by
    \begin{equation} O(N) + O\left( \frac{K_N\log N}{N^3} \right) - c\frac{K_N}{N^2}, \end{equation} which tends to $-\infty$ because $K_N=2^{rN+o(N)}$ with $r>0$. Hence, for all sufficiently large $N$, there exists an isometry $W_N$ such that \begin{equation} g_{N,a}(W_Nx) \leq \epsilon_N^2+\frac{1}{2N} 
    \end{equation} simultaneously for every $a\in\mathcal A_N$ and $x\in\mathcal M_N$.
    Fix such an isometry and define $\mathcal C'_N := \operatorname{im}(W_N)$. For any unit vector $\ket{y} \in \mathbb C^{K'_N}$, choose $x \in \mathcal M_N$ such that $\|x-y\|_2 \leq 1/(8N)$. Using the Lipschitz bound and the fact that $W_N$ is an isometry,
    \begin{align} g_{N,a}(W_Ny) &\leq g_{N,a}(W_Nx) + 4\|W_N(x-y)\|_2 \\ &\leq \epsilon_N^2 + \frac{1}{2N} + \frac{1}{2N} \\ &= \epsilon_N^2+\frac{1}{N}. \end{align}
    This holds simultaneously for every $a\in\mathcal A_N$ and every unit vector in $\mathcal C'_N$. Equivalently, \begin{equation} \min_{a\in\mathcal A_N} \min_{\substack{|\psi\rangle\in\mathcal C'_N\\ \langle\psi|\psi\rangle=1}} \langle\psi| \Lambda_{N,a}(|\psi\rangle\langle\psi|) |\psi\rangle \geq 1-\epsilon_N^2-\frac{1}{N}, \end{equation} which proves the lemma.
\end{proof}
To connect the pure state channel fidelity to entanglement fidelity, we use the following result of~\cite{barnum_quantum_2000}:
\begin{lemma}[Theorem 2 in~\cite{barnum_quantum_2000}]\label{lem: barnum}
    Let $\cC \subset \C^K$ be a subspace and let $\Lambda: \cL(\C^K) \to \cL(\C^K)$ be a quantum channel. Suppose that
    \begin{align}
        \min_{\substack{\ket{\psi} \in \cC \\ \langle\psi|\psi\rangle = 1}} \bra{\psi} \Lambda(\ketbra{\psi}) \ket{\psi} \ge 1 - \eps^2.
    \end{align}
    Then
    \begin{align}
        \min_{\rho \in \cD(\cC)} F_e(\rho, \Lambda) \ge 1 - \frac{3}{2} \eps^2.
    \end{align}
\end{lemma}
\textit{Proof of Theorem S1.} For each $a \in \mathcal A_N$, define the recovered logical channel
\begin{equation}
    \Lambda_{N,a}
:=
\mathcal R_{N,a}
\circ
\mathcal N_{N,a}
\circ
\mathcal E_N,
\end{equation}
and let
\begin{equation}
    \epsilon_N
:=
\max_{a\in\mathcal A_N}
\left(
1-
F_e\left(
\frac{\id_{K_N}}{K_N},
\Lambda_{N,a}
\right)
\right)^{1/2}.
\end{equation}
By hypothesis, $\epsilon_N \to 0$. The case $r= 0$ is trivial; assume $r>0$. Henceforth suppose $r > 0$. By Lemma~\ref{lemm:haar_subcode_space}, there exists a single subspace $\mathcal C'_N\subset\mathbb C^{K_N}$, independent of $a\in\mathcal A_N$, with dimension $K'_N
=
\left\lfloor
\frac{K_N}{N^3}
\right\rfloor$, such that \begin{equation}
\min_{a\in\mathcal A_N}
\min_{\substack{|\psi\rangle\in\mathcal C'_N\\
\langle\psi|\psi\rangle=1}}
\langle\psi|
\Lambda_{N,a}(|\psi\rangle\langle\psi|)
|\psi\rangle
\geq
1-\epsilon_N^2-\frac{1}{N}.
\end{equation}
Applying Lemma~\ref{lem: barnum} to each channel $\Lambda_{N,a}$ on the same subspace $\mathcal C'_N$ gives 
\begin{equation}
\min_{a\in\mathcal A_N}
\min_{\rho\in\mathcal D(\mathcal C'_N)}
F_e(\rho,\Lambda_{N,a})
\geq
1-
\frac{3}{2}
\left(
\epsilon_N^2+\frac{1}{N}
\right).
\end{equation}
Identifying $\mathcal C'_N$ isometrically with $\mathbb C^{K'_N}$ and
restricting the encoding and recovery maps accordingly, with the
inverse isometry extended arbitrarily to a CPTP map outside
$\mathcal C'_N$, cannot decrease the above entanglement fidelities.
Hence,
\begin{equation}
\max_{a\in\mathcal A_N}
\left\|
\mathcal R'_{N,a}
\circ
\mathcal N_{N,a}
\circ
\mathcal E'_N
-
\operatorname{id}_{K'_N}
\right\|_\diamond
\leq
2\left(
\frac{3}{2}
\left(
\epsilon_N^2+\frac{1}{N}
\right)
\right)^{1/2} = \sqrt{6}\left(\epsilon^2_N + \frac{1}{N}\right)^{1/2}
\longrightarrow0.
\end{equation}  
Finally,
\begin{equation}
\frac{1}{N}\log_2K'_N
=
\frac{1}{N}\log_2K_N
-
\frac{3\log_2N}{N}
+
o(1)
\longrightarrow r,
\end{equation}
so $\mathcal C'_N$ has the same asymptotic rate as the original code. 

\subsection{Average-Case Approximate Quantum Error Correction from Decoupling} \label{sec:aqec_decoupling}

Theorem~\ref{thm: worst_to_avg} justifies the study of quantum thermalization through the lens of average-case AQEC, rather than worst-case. As noted, this leads to a more computationally simple approach to analyze and prove error correction. In this section, we connect average-case AQEC to the \emph{decoupling} condition studied in the main text~\cite{hayden_decoupling_2008}. Specifically, we demonstrate how to derive a lower bound on the entanglement fidelity of the maximally mixed state based on a decoupling analysis of the purified state after the noise channel. Notably, this enables us to prove average-case AQEC without ever explicitly constructing a decoding channel.

In what follows, we let $Q \cong \C^K$ and $S \cong \C^D$ denote the input and output spaces to the encoding map $\cE$, we let $R \cong \C^K$ denote the purification space for the input density matrix $\id/K$, and we let $E$ denote the dilation space for the noise channel $\cN$. 
We denote the maximally entangled state between $R$ and $Q$ as $\ket{\Psi_{RQ}}$; this is the purification of the input density matrix $\id/K$.
We let $U_{SE}$ be the Stinespring dilation of $\cN$. We define the state $| \Psi_{RSE}^{\text{noisy}} \rangle$ 
as the final state on the full system $RSE$.
We denote the partial traces of $\ket{\Psi}$ and $| \Psi^{\text{noisy}} \rangle$ by $\sigma$ and $\sigma^{\text{noisy}}$, respectively, with subscript indicating the reduced space on which the state is defined.

\begin{figure}
    \centering
    \includegraphics[width=.85\linewidth
    ]{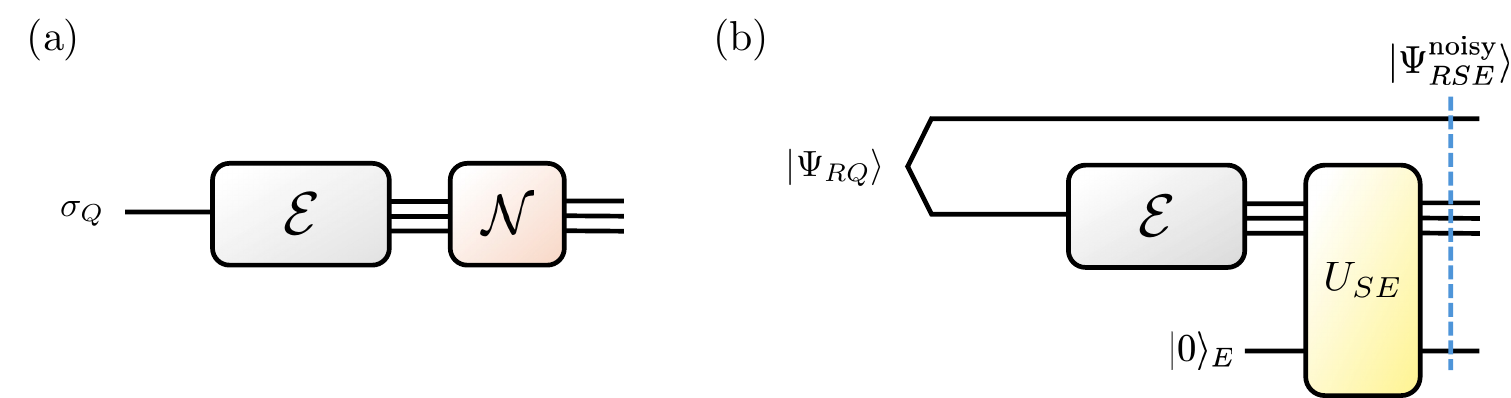}
    \caption{Setup for deriving average-case approximate quantum error correction from decoupling. (a) An encoding of the maximally mixed state $\sigma_Q = \mathbbm{1} / K$ (i.e., a uniform mixture of all codewords) subjected to a noise channel $\cN$. (b) Extending panel (a) to a larger Hilbert space by introducing a purification $| \Psi_{RQ} \rangle$ of $\sigma_Q$ and Stinespring dilation $U_{SE}$ of the noise channel $\cN$.}
    \label{fig:decoupling_diagram}
\end{figure}

We demonstrate the following connection between decoupling and average-case AQEC:

\begin{lemma}[Adapted from~\cite{schumacher_approximate_2001}]
    Suppose that the quantum mutual information of $| \Psi_{RSE}^{\mathrm{noisy}} \rangle$ satisfies
    \begin{align} 
        I(R: E) \le \eps^4/2.
    \end{align}
    Then there exists a decoding channel $\cR: \cL(\C^D) \to \cL(\C^K)$ such that $(\cE, \cR)$ is an $\eps$-approximate average-case quantum error correcting code for noise channel $\cN$. 
\end{lemma}
\begin{proof}
    First, by the quantum Pinsker inequality, we may bound the trace distance between $\sigma_{RE}^{\text{noisy}}$ and $\sigma_{R}^{\text{noisy}} \otimes \sigma_{E}^{\text{noisy}}$ as follows:
    \begin{align}
         \frac{1}{2} \| \sigma_{RE}^{\text{noisy}} - \sigma_{R}^{\text{noisy}} \otimes \sigma_{E}^{\text{noisy}} \|_1 &\le \left( \frac{\ln(2)}{2} S\big(\sigma_{RE}^{\text{noisy}} \| \sigma_{R}^{\text{noisy}} \otimes \sigma_{E}^{\text{noisy}}\big) \right)^{1/2}
         \\
         &= \left( \frac{\ln(2)}{2} I(R:E) \right)^{1/2}
         \\
         &\le \eps^2/2.
    \end{align}
    By the Fuchs--van de Graaf inequalities~\cite{Watrous2018TQI}, it follows that the fidelity (defined with a ``square'') between $\sigma_{RE}^{\text{noisy}}$ and $\sigma_{R}^{\text{noisy}} \otimes \sigma_{E}^{\text{noisy}}$ satisfies
    \begin{align}\label{eq:fvdg}
        F(\sigma_{RE}^{\text{noisy}}, \sigma_{R}^{\text{noisy}} \otimes \sigma_{E}^{\text{noisy}}) \ge (1 - \eps^2/2)^2.
    \end{align}
    By Uhlmann's theorem, there exists a purification $| \Psi_{RSE}^{\mathrm{decoup}} \rangle$ of $\sigma_{R}^{\text{noisy}} \otimes \sigma_{E}^{\text{noisy}}$ such that
    \begin{align}\label{eq:uhlmann}
         F(\sigma_{RE}^{\text{noisy}}, \sigma_{R}^{\text{noisy}} \otimes \sigma_{E}^{\text{noisy}}) = | \langle \Psi_{RSE}^{\mathrm{decoup}} | \Psi_{RSE}^{\mathrm{noisy}} \rangle|^2.
    \end{align}
    Since $\sigma_{R}^{\text{noisy}} \otimes \sigma_{E}^{\text{noisy}}$ is perfectly decoupled, there exists a channel $\cR: \cL(\C^D) \to \cL(\C^K)$ acting on $S$ alone such that
    \begin{align}
        \cR \otimes \mathrm{id}_R(\sigma_{RS}^{\mathrm{decoup}}) = \sigma_{RQ}.
    \end{align}
    Now consider the entanglement fidelity of $\sigma_Q$ under $\cR \circ \cN \circ \cE$. We have
    \begin{align}
        F_e(\sigma_Q, \cR \circ \cN \circ \cE) &= | \langle \Psi_{RQ}| \cR \circ \cN \circ \cE \otimes \mathrm{id}_R (\sigma_{RQ}) | \Psi_{RQ} \rangle|
        \\
        &= F(\sigma_{RQ}, \cR \circ \cN \circ \cE \otimes \mathrm{id}_R (\sigma_{RQ}))
        \\
        &= F(\cR \otimes \mathrm{id}_R(\sigma_{RS}^{\mathrm{decoup}}), \cR \circ \cN \circ \cE \otimes \mathrm{id}_R (\sigma_{RQ}))
        \\
        &\ge F(\sigma_{RS}^{\mathrm{decoup}}, \cN \circ \cE \otimes \mathrm{id}_R (\sigma_{RQ}))
        \\
        &\ge F(\sigma_{RSE}^{\mathrm{decoup}}, \sigma^{\mathrm{noisy}}_{RSE}))
        \\
        &\ge (1-\eps^2/2)^2,
    \end{align}
    where the first and second inequalities follow from monotonicity of fidelity under $\cR$ and $\tr_E$, respectively, and the final line by Eq.~\eqref{eq:fvdg} and Eq.~\eqref{eq:uhlmann}. Rearranging, we obtain
    \begin{align}
        \left(1 - F_e(\sigma_Q, \cR \circ \cN \circ \cE) \right)^{1/2} \le (1-(1-\eps^2/2)^2)^{1/2} \le \eps,
    \end{align}
    as desired.
\end{proof}

\subsection{Optimality for Bounded Entropy Codes}\label{sec: entropic_singleton}
Having established approximate correctability from decoupling, we next ask whether the resulting rate-distance tradeoff is optimal. For arbitrary quantum codes, the relevant converse is the quantum Singleton bound. When the relevant reduced states of the maximally mixed code state have entropy below the maximum allowed by their Hilbert-space dimension, the dimension-based Singleton bound can be loose. The entropic quantum Singleton bound incorporates these actual reduced-state entropies and can therefore provide a tighter converse. We begin by recalling the standard quantum Singleton bound.
\begin{theorem}[Quantum Singleton bound \cite{knill_theory_1997}] Let $\mathcal C$ be an $[[N,k,d]]$ quantum error-correcting code where $N$ is the number of physical qubits, $k = \log_2 \dim \mathcal C$ is the number of logical qubits, and $d$ is the code distance. Then \begin{equation}
    k \le N-2(d-1).
\end{equation}
Equivalently, if the code corrects arbitrary erasures of $d-1$ qubits, writing the encoding rate as $r = k/N$ and correctable erasure fraction as $p = (d-1)/N$, 
\begin{equation}
    p \leq \frac{1-r}{2}.
\end{equation}
\end{theorem}

For standard quantum codes, optimality is typically determined through approximate saturation of this bound, together with other constraints such as the quantum Hamming bound \cite{ekert_error_1996}. However, when restricting to families of codes with bounded entropy, these conventional bounds generally become loose and are no longer the relevant limits. The usual Singleton bound only uses the
local Hilbert-space dimension and does not account for the fact that
the relevant reduced states may occupy a much smaller entropic
subspace. The entropic Singleton bound of
Grassl--Huber--Winter~\cite{grassl_entropic_2022} refines this
statement by replacing the maximal entropy of erased regions by their
actual entropy.

\begin{theorem}[Entropic Quantum Singleton Bound (Adapted from \cite{grassl_entropic_2022})]\label{thm:entropic_singleton} Let $\mathcal C$ be an $[[N,k,d]]$ quantum error-correcting code with projector $P_{\mathcal C}$ and let $\bar \rho_S = \frac{1}{K}P_{\mathcal C}$ 
    be the maximally mixed state on the code. Let $\ket{\Omega_{RS}}$ be a purification of $\bar \rho_S$, so that $S(\Omega_R)=\log_2 K=k$. For any subsystem $A$, write $S(A)\coloneqq S(\Omega_A)$. Suppose that $\mathcal C$ exactly corrects every erasure of $d-1 < N/2$ qubits. Then, 
    \begin{equation} \label{eq:entropic_singleton_avg}
\log_2 K
\le
\frac{N-2(d-1)}{d-1}
\overline S_{d-1},
\qquad
\overline S_{d-1}
\coloneqq
\mathbb E_{|A|=d-1} S(A),
\end{equation}
where the average is over all regions $A$ of size $d-1$ \cite{grassl_entropic_2022}.

Moreover, for any two disjoint correctable erasure regions $A,C$ with $|A|=|C|=d-1$, defining $B=[N]\setminus(A\cup C)$,
one has the pairwise bound
\begin{equation}\label{eq:entropic_singleton_pairwise}
\log_2 K
\le
\min_{\substack{A\cap C=\emptyset\\ |A|=|C|=d-1}}
S\left([N]\setminus(A\cup C)\right).
\end{equation}
For $p=(d-1)/N$, this becomes a bound on the achievable rate $r=\log_2K/N$.
\end{theorem}

The inequality Eq.~\eqref{eq:entropic_singleton_avg} is the averaged entropic Singleton bound of Ref.~\cite{grassl_entropic_2022}. This averaged bound provides the natural benchmark when the entropies of erasure regions of a fixed size are spatially uniform, or asymptotically comparable. For general bounded-entropy codes, however, correctability is a worst-case requirement: every erasure of size $d-1$ must be correctable. Consequently, the Singleton argument may be applied separately to every pair of disjoint correctable regions $A$ and $C$. Exact correctability implies $I(E:R) = 0$, and purity together with subadditivity and taking the minimum amongst such pairs gives Eq.~\eqref{eq:entropic_singleton_pairwise}. This pairwise form gives the operative constraint when the reduced-state entropies are spatially nonuniform.

For the approximate codes studied in this work, we use the entropic Singleton bound as an asymptotic converse. Namely, for a family of codes indexed by system size $N$, vanishing decoupling error implies that the mutual-information corrections to the exact Singleton argument are $o(N)$. Hence the bound holds at the level of entropy densities and rates in the thermodynamic limit, up to terms that vanish as $N \rightarrow \infty$. We will use the pairwise form of the entropic Singleton bound, which involves a minimum over disjoint correctable erasure regions. Thus, the converse is controlled by the least-entropic complementary region $B=[N]\setminus(A\cup C)$ of the relevant size; possible boundary- or geometry-dependent entropy increases for other choices of $A,C$ do not weaken this converse. 

It will be convenient to fix a canonical purification of the maximally
mixed code state, written explicitly in terms of the encoding isometry. Let $V:\mathbb C^{K}\to\mathcal H_S$ be an encoding isometry, let
$Q\cong\mathbb C^{K}$ denote the logical space, and let $R\cong\mathbb C^{K}$ be a reference system.  Writing
\begin{equation}
  |\omega_K\rangle_{RQ}
  :=
  \frac{1}{\sqrt{K}}\sum_{i=1}^{K}|i\rangle_R\otimes|i\rangle_Q
\end{equation}
for the maximally entangled state on $RQ$, we define the encoded state
\begin{equation}
  |\Omega(V)\rangle_{RS}:=(\mathbbm 1_R\otimes V)\,|\omega_K\rangle_{RQ},
  \qquad
  \Omega(V):=|\Omega(V)\rangle\!\langle\Omega(V)| .
  \label{eq:encoded-purification}
\end{equation}
Reduced states of $\Omega(V)$ are denoted by the same symbol with a
subscript indicating the retained subsystem, so that
\begin{equation}
  \Omega_S(V)=\frac{VV^\dagger}{K}=\overline\rho,
  \qquad
  \Omega_R(V)=\frac{\mathbbm 1_R}{K},
  \qquad
  S(R)_{\Omega}=\log_2 K ,
\end{equation}
and $|\Omega(V)\rangle_{RS}$ is therefore a purification of
$\overline\rho$ of the form appearing in
Theorem~\ref{thm:entropic_singleton}.  For an erased region
$E\subseteq[N]$ with complement $B=[N]\setminus E$, we write
$\Omega_{RE}(V)=\operatorname{tr}_B\Omega(V)$ and define the decoupling
error
\begin{equation}
  \Delta_E(V):=
  \left\|
    \Omega_{RE}(V)-\frac{\mathbbm 1_R}{K}\otimes\Omega_E(V)
  \right\|_1 .
  \label{eq:decoupling-error}
\end{equation}
Entropies and mutual informations carrying a subscript $\Omega$, such as
$S(A)_\Omega$ and $I(R:A)_\Omega$, are evaluated in $\Omega(V)$. We now specialize this converse to uniform volume-law codes, for which the pairwise entropic Singleton bound reduces to a direct constraint on the rate and correctable erasure fraction. 

\begin{corollary}[Volume-law entropic Singleton converse]
\label{corr:entropic_singleton_volume_law}
Let $V_N:\mathbb C^{K_N}\to(\mathbb C^2)^{\otimes N}$ be a sequence
of encoding isometries, and define $\overline\rho_N
:=
\frac{V_NV_N^\dagger}{K_N}$ and $r_N
:=
\frac{\log_2K_N}{N}$. Let $e_N<N/2$ and $p_N:=e_N/N$. Suppose that
\begin{equation}
\max_{|E|=e_N}
\left\|
\Omega_{RE}(V_N)
-
\frac{\mathbbm 1_R}{K_N}\otimes\Omega_E(V_N)
\right\|_1
\longrightarrow0,
\end{equation}
and that, for some $s_N>0$,
\begin{equation}
\min_{\substack{A\cap C=\emptyset\\
|A|=|C|=e_N}}
S\left(
(\overline\rho_N)_{[N]\setminus(A\cup C)}
\right)
\leq
s_N(N-2e_N)+o(N).
\label{eq:volume-law-entropy-upper-bound}
\end{equation}
If $r_N\longrightarrow r$, $p_N\longrightarrow p$, $s_N\longrightarrow s>0$, then
\begin{equation}
p
\leq
\frac{s-r}{2s}.
\label{eq:vol_entropic_Singleton}
\end{equation}
\end{corollary}

\begin{proof}
Let $A$ and $C$ be disjoint regions of size $e_N$, and set
$D=[N]\setminus(A\cup C)$. Since $\Omega_{RACD}$ is pure,
\begin{equation}
\begin{aligned}
I(R:A)_\Omega+I(R:C)_\Omega
&=
2S(R)_\Omega+S(A)_\Omega+S(C)_\Omega -
S(CD)_\Omega-S(AD)_\Omega \\
&\geq
2S(R)_\Omega-2S(D)_\Omega,
\end{aligned}
\end{equation}
where the inequality follows from subadditivity. Therefore,
\begin{equation}
\log_2K_N
=
S(R)_\Omega
\leq
S\left((\overline\rho_N)_D\right)
+
\frac{1}{2}I(R:A)_\Omega
+
\frac{1}{2}I(R:C)_\Omega.
\label{eq:approximate-pairwise-singleton}
\end{equation}
The assumed decoupling condition and continuity of entropy imply, uniformly over all $E$ of size $e_N$, $I(R:E)_\Omega=o(N)$. Taking the minimum of Eq.~\eqref{eq:approximate-pairwise-singleton} over disjoint $A,C$ and applying Eq.~\eqref{eq:volume-law-entropy-upper-bound} gives \begin{equation}
r_N N
\leq
s_N(N-2e_N)+o(N)
=
s_NN(1-2p_N)+o(N).
\end{equation}
Dividing by $N$ and taking $N\to\infty$ yields $r\leq s(1-2p)$, which is equivalent to
Eq.~\eqref{eq:vol_entropic_Singleton}.
\end{proof}

\section{Scrooge Achievability Bounds}\label{sec:scrooge_achieve}

In the previous section, we showed that, for codes families whose maximally mixed code states have bounded entropy, the entropic quantum Singleton bound can be tighter than the usual dimension-based quantum Singleton bound. A natural question is whether this tighter bound can be achieved. We prove that it can, by constructing codes which naturally account for constraints, such as energy or charge conservation, which lead to bounded entropy. These codes are built by drawing codewords randomly from the \emph{Scrooge ensemble} of a density matrix $\rho$ encoding the constraints, then orthogonalizing the states to form a valid coding isometry. The Scrooge ensemble is a natural choice for this code construction, since it has the smallest accessible information among all state ensembles whose first moment is $\rho$. Furthermore, for maximally mixed $\rho$, the Scrooge ensemble is simply the Haar ensemble, which is known to yield an asymptotically optimal code construction in the constraint-free setting~\cite{ma_haar_2025}.

We prove that these so-called \emph{Scrooge random codes} satisfy a decoupling condition; by the analysis of Sec.~\ref{sec:aqec_decoupling}, this implies approximate recoverability. After proving a decoupling bound which holds for any choice of average state $\rho$, we prove that two further assumptions imply a random Scrooge code saturates the entropic quantum Singleton bound. The first is that the rate $r$ of the code is at most the min entropy-density of $\rho$; otherwise, the support of $\rho$ is not large enough to accommodate the codespace. The second is a physically-motivated spectral concentration assumption on the reduced states of $\rho$. In the following Sec.~\ref{sec:tightness}, we provide explicit examples of $\rho$ which satisfy these conditions as well as bounded entropy, proving that the entropic quantum Singleton bound is tight for all values of the rate, erasure fraction, and entropy density. Finally, we discuss the case that $\rho$ is the Gibbs state of a local Hamiltonian, providing an analytically tractable model for the thermalization codes discussed in the main text. In all cases, we prove that the decoupling condition holds simultaneously for all erasures of a given size, thus proving worst-case error correction.

\subsection{Scrooge Random Codes}\label{sec:scrooge_codes_def}

We begin by formally defining the Scrooge ensemble:

\begin{defn}[Scrooge ensemble]
\label{def:scrooge}
Let $\rho$ be a density matrix of rank $d_\rho := \mathrm{rank}(\rho)$ on $\cL(\cH)$ for some Hilbert space $\cH$.  Draw a unit vector
$|\psi\rangle\in\mathrm{supp}(\rho)$ from the probability density
$d_\rho\,\langle\psi|\rho|\psi\rangle$
relative to normalized Haar measure, and set
\begin{equation}
  |\phi\rangle
  =\frac{\sqrt\rho\,|\psi\rangle}
  {\sqrt{\langle\psi|\rho|\psi\rangle}}.
  \label{eq:scrooge-definition}
\end{equation}
The resulting distribution of pure states is denoted $\Scr(\rho)$
\cite{jozsa_lower_1994}.
\end{defn}
\noindent
The Scrooge ensemble is also known as the \emph{Gaussian adjusted projected }(GAP) ensemble \cite{goldstein_distribution_2006, reimann_typicality_2008}, with an equivalent construction as follows. Let
$|g\rangle$ be a standard complex Gaussian vector in $\mathrm{supp}(\rho)$, put
$|y\rangle=\sqrt\rho|g\rangle$ and $T=\langle y|y\rangle$. Reweight the Gaussian measure on $\ket{g}$ by $T$,
and project $|y\rangle$ to the unit sphere.  In this form it is immediate that
\begin{equation}
  \underset{|\phi\rangle\sim\Scr(\rho)}{\mathbb E}\dyad{\phi}
  =\mathbb E_g\,|y\rangle\!\langle y|=\rho.
  \label{eq:scrooge-first-moment}
\end{equation}
Draw $K$ independent samples $|\phi_1\rangle,\ldots,|\phi_K\rangle\sim\Scr(\rho)$ and define
\begin{equation}
  \Phi=\begin{pmatrix}|\phi_1\rangle&\cdots&|\phi_K\rangle\end{pmatrix},
  \qquad
  G=\Phi^\dagger\Phi.
\end{equation}
For $K\le\rank\rho$ -- which we will always assume -- linear independence holds almost surely. In this case, $G$ is invertible, and symmetric orthogonalization~\cite{lowdin_nonorthogonality_1970} defines an encoding isometry $V$ according to
\begin{equation}
  V=\Phi G^{-1/2},
  \qquad V^\dagger V=\id_K.
  \label{eq:lowdin}
\end{equation}
The corresponding code projector and unorthogonalized operator are
\begin{equation}
  P=VV^\dagger=\Phi G^{-1}\Phi^\dagger,
  \qquad
  \widetilde P=\Phi\Phi^\dagger
  =\sum_{i=1}^K\ketbra{\phi_i}.
  \label{eq:P-and-Ptilde}
\end{equation}

We first show that the Gram matrix $G$ is close to identity with high probability, so that the normalization can be essentially ignored. For fixed $\eta >0$, the failure probability is superexponentially small whenever $K\|\rho\|_{\mathrm{op}}$ is exponentially small, where $\|\rho\|_{\mathrm{op}}$ denotes the largest eigenvalue of $\rho$. More generally, the relevant concentration parameter is $\eta^2/(K\|\rho\|_{\mathrm{op}})$. We will use the matrix Bernstein inequality:

\begin{fact}[Matrix Bernstein inequality~\cite{tropp_user-friendly_2012}]
\label{fact:matrix-bernstein}
Let $X_1,\ldots,X_m$ be independent mean-zero Hermitian random matrices of size $K\times K$.  Suppose
that for all integers $p\ge2$ and both signs,
\begin{equation}
  \mathbb E[(\pm X_j)^p]
  \preceq\frac{p!}{2}R^{p-2}A_j^2
\end{equation}
for some $R>0$ and positive semidefinite $A_j^2$.  Letting
$\sigma^2=\opnorm{\sum_jA_j^2}$, for all $t \ge 0$
\begin{equation}
  \mathrm{Pr}\!\left\{ \Big\| \sum_jX_j \Big\|\ge t\right\}
  \le2K\exp\!\left(-\frac{t^2/2}{\sigma^2+Rt}\right).
  \label{eq:matrix-bernstein}
\end{equation}
\end{fact}

We now prove a concentration inequality for the Gram matrix of independent Scrooge random vectors, thereby lower-bounding the probability of the ``good Gram event'' $\| G - \id_K \| \le \eta$.

\begin{lemma}[Gram-matrix concentration]
\label{lem:gram}
Let $\lambda=\opnorm{\rho}$ and let $G$ be the Gram matrix of $K$ independent
$\Scr(\rho)$ samples.  There are universal constants $c,C>0$ such that, for
$0<\eta\le1/2$,
\begin{equation}
  p_\eta \coloneqq \mathrm{Pr}\{\opnorm{G-\id_K}>\eta\}
  \le C\sqrt K\,
  \exp\!\left(\frac{K\lambda}{2}
  -c\frac{\eta^2}{K\lambda}\right).
  \label{eq:gram-tail}
\end{equation}
\end{lemma}

\begin{proof}
Diagonalize $\rho=\sum_{a=1}^{d_\rho} p_a\ketbra{a}$ and let
$|g_1\rangle,\ldots,|g_K\rangle$ be independent standard complex Gaussian vectors.  Define
\begin{equation}
  \widetilde G_{ij}=\langle g_i|\rho|g_j\rangle.
\end{equation}
Writing $v_a=((g_1)_a,\ldots,(g_K)_a)^\top\in\mathbb C^K$ for $(g_i)_a = \braket{g_i}{a}$, one has
\begin{equation}
  \widetilde G-\id_K
  =\sum_{a=1}^{d_\rho} p_a(v_av_a^\dagger-\id_K).
  \label{eq:weighted-wishart}
\end{equation}
The summands are independent and mean zero.  By unitary invariance,
$\mathbb E(vv^\dagger-\id)^p$ is a scalar multiple of the identity. Taking the trace, if
$R_v=\|v\|_2^2\sim\mathrm{Gamma}(K,1)$, then we see that
\begin{equation}
  \mathbb E(vv^\dagger-\id)^p
  =\frac{\mathbb E(R_v-1)^p+(K-1)(-1)^p}{K}\,\id_K.
\end{equation}
We can bound these moments directly. Since
$\mathbb E R_v^p=p!\binom{K+p-1}{p}$, one has
\begin{equation}
  \mathbb E R_v^p\le p!(4K)^p.
  \label{eq:gamma-moment-bound}
\end{equation}
Indeed, if $p<K$, then $p!\binom{K+p-1}{p}\le(K+p)^p\le(2K)^p$; if $p\ge K$, then
$\binom{K+p-1}{p}\le2^{K+p-1}\le4^p$.  Using
$(R_v+1)^p\le2^{p-1}(R_v^p+1)$ (by Jensen's inequality), it follows that, for a universal $C_0$,
\begin{equation}
  \frac{1}{K}\left(\mathbb E(R_v+1)^p+K-1\right)
  \le p!\,C_0^pK^{p-1}.
  \label{eq:centered-gamma-moment}
\end{equation}
Since $|R_v-1|^p\le(R_v+1)^p$, the absolute value of the scalar coefficient
of $\id_K$ is bounded by the left-hand side of Eq.~\eqref{eq:centered-gamma-moment}.  Because $p_a^{p-2}\le\lambda^{p-2}$, we can choose sufficiently large universal
constants $C_1,C_2$ to satisfy the Bernstein moment condition with
$R=C_1K\lambda$ and $A_a^2=C_2Kp_a^2\id_K$.  Thus,
$\sigma^2\le C_2K\tr\rho^2\le C_2K\lambda$, and
Fact~\ref{fact:matrix-bernstein} gives, for $0<t\le1$,
\begin{equation}
  \Pr_0\{\| \widetilde G-\id_K \|>t\}
  \le2K\exp\!\left(-c_0\frac{t^2}{K\lambda}\right).
  \label{eq:gaussian-gram-tail}
\end{equation}
Here $\Pr_0$ denotes the usual, unweighted Gaussian law.

Let $D=\mathrm{diag}(\widetilde G)$ and define
\begin{equation}
  |\varphi_i\rangle
  =\frac{\sqrt\rho|g_i\rangle}
  {\sqrt{\langle g_i|\rho|g_i\rangle}}.
\end{equation}
The Gram matrix of these normalized vectors is
$G=D^{-1/2}\widetilde GD^{-1/2}$.  If
$\| \widetilde G-\id \| \le\eta/4$, then
$\opnorm{D-\id}\le\eta/4$ and
\begin{equation}
  \opnorm{G-\id}
  \le\frac{\eta/2}{1-\eta/4}\le\eta,
  \label{eq:normalize-gram}
\end{equation}
by submultiplicativity of operator norm.

The estimate Eq.~\eqref{eq:gaussian-gram-tail} was derived under the unweighted Gaussian law, in which the directions $|\psi_i \rangle = |g_i\rangle /\||g_i\rangle\|_2$ are Haar random. The Scrooge samples instead have directions drawn from the tilted density of Definition~\ref{def:scrooge}, so we must transfer the bound between the two measures. Writing $q_i = \langle \psi_i|\rho |\psi_i\rangle$, the joint likelihood ratio of the Scrooge directions relative to independent Haar directions is $L = \prod_{i=1}^K d_\rho q_i$. For a Haar-random unit vector,
\begin{equation}
  \mathbb E_0[q_i^2]
  =\frac{1+\tr\rho^2}{d_\rho(d_\rho+1)},
\end{equation}
so
\begin{equation}
  \mathbb E_0[L^2]
  =\left(\frac{d_\rho(1+\tr\rho^2)}{d_\rho+1}\right)^K
  \le\exp(K\lambda).
\end{equation}
Cauchy--Schwarz, Eq.~\eqref{eq:gaussian-gram-tail}, and
 Eq.~\eqref{eq:normalize-gram} now prove Eq.~\eqref{eq:gram-tail}, after changing universal constants.
\end{proof}

The following comparison will be used repeatedly. On the event
$\opnorm{G-\id_K}\le\eta$,
\begin{equation}
  \frac{1}{1+\eta}\widetilde P
  \preceq P
  \preceq
  \frac{1}{1-\eta}\widetilde P.
  \label{eq:loewner-comparison}
\end{equation}
Consequently, for any positive linear map $\mathcal T$,
\begin{equation}
  \frac{1}{1+\eta}\mathcal T(\widetilde P)
  \preceq
  \mathcal T(P)
  \preceq
  \frac{1}{1-\eta}\mathcal T(\widetilde P).
  \label{eq:loewner-comparison-positive-map}
\end{equation}
Since all three operators are positive semidefinite, Weyl monotonicity implies
the corresponding comparison of their eigenvalues. Therefore,
\begin{equation}
  \frac{1}{(1+\eta)^2}
  \tr\!\left[\mathcal T(\widetilde P)^2\right]
  \le
  \tr\!\left[\mathcal T(P)^2\right]
  \le
  \frac{1}{(1-\eta)^2}
  \tr\!\left[\mathcal T(\widetilde P)^2\right].
  \label{eq:loewner-square-trace-comparison}
\end{equation}

\subsection{Scrooge Second Moment}

We now derive the exact second moment of the Scrooge ensemble in a form suited to the decoupling calculation. This result was previously derived in componentwise form for the GAP ensemble in~\cite{reimann_typicality_2008} and follows as a special case of the general moment formulas in~\cite{mark_maximum_2024}.

\begin{lemma}[Exact Scrooge second moment]
\label{lem:exact-scrooge-second}
Let
$\rho=\sum_{a=1}^{d_\rho} p_a\ketbra{a}$ with $p_a>0$, and let
$\lambda=\max_a p_a$.  Then
\begin{align}
  \chi^{(2)}_{\Scr(\rho)}
  &:={\mathbb E}_{\ket{\phi} \sim\Scr(\rho)}\big[\ketbra\phi^{\otimes2}\big] =\sum_{a,b=1}^{d_\rho} p_ap_bK_{ab}
  \left(
    \ketbra a\otimes\ketbra b
    +\ketbra{a}{b}\otimes\ketbra{b}{a}
  \right),
  \label{eq:exact-scrooge-second}
\end{align}
where
\begin{equation}
  K_{ab}=\int_0^\infty
  \frac{\mathrm{d} x}
  {(1+xp_a)(1+xp_b)\prod_{c=1}^{d_\rho}(1+xp_c)}.
  \label{eq:Kab}
\end{equation}
Moreover, if $\lambda<1$, then
\begin{equation}
  0<K_{ab}\le\frac{1}{1-\lambda}.
  \label{eq:Kab-bound}
\end{equation}
\end{lemma}

\begin{proof}
We use the adjusted Gaussian representation of the Scrooge ensemble.  Let $z_a$ be independent standard complex Gaussian
variables normalized by $\mathbb E_0 [z_a\bar{z}_a] = \delta_{ab}$ and let $\mathbb E_0$ denote expectation under this unweighted Gaussian law. Let $|y\rangle=\sum_a\sqrt{p_a}z_a|a\rangle$, and 
$T=\sum_ap_a|z_a|^2$.  Then
\begin{equation}
  \chi^{(2)}_{\Scr(\rho)}
  =\mathbb E_0\left[\frac{\ketbra y^{\otimes2}}{T}\right].
\end{equation}
Insert $T^{-1}=\int_0^\infty e^{-xT}\mathrm{d} x$.  If
$G(x)=\prod_c(1+xp_c)^{-1}$, independence of the Gaussian coordinates and Wick's rule give
\begin{align}
  &\mathbb E\!
  \left[z_i\overline z_jz_k\overline z_l\,e^{-xT}\right] =G(x)\left(
  \frac{\delta_{ij}\delta_{kl}}
       {(1+xp_i)(1+xp_k)}
  +\frac{\delta_{il}\delta_{kj}}
       {(1+xp_i)(1+xp_k)}
  \right).
  \label{eq:tilted-wick}
\end{align}
The first pairing produces
$\ketbra i\otimes\ketbra k$, while the second produces
$\ketbra{i}{k}\otimes\ketbra{k}{i}$.  Integrating over $x$ therefore gives
 Eq.~\eqref{eq:exact-scrooge-second}, Eq.~\eqref{eq:Kab}.

For the bound, again set
$G(x)=\prod_c(1+xp_c)^{-1}$.  Since $K_{ab}\le\int_0^\infty G(x)\mathrm{d} x$, it is enough to
bound the latter integral.  For fixed $x\ge0$, the function
$t\mapsto t^{-1}\log (1+xt)$ is decreasing on $(0,\infty)$.  Since $p_c\le\lambda$ and
$\sum_cp_c=1$,
\begin{equation}
  \sum_c\log (1+xp_c)
  \ge\frac{1}{\lambda}\log (1+x\lambda).
\end{equation}
Hence
\begin{equation}
  G(x)\le(1+x\lambda)^{-1/\lambda},
\end{equation}
and therefore
\begin{equation}
  \int_0^\infty G(x)\mathrm{d} x
  \le\int_0^\infty(1+x\lambda)^{-1/\lambda}\mathrm{d} x
  =\frac{1}{1-\lambda}.
\end{equation}
\end{proof}

Let $\cH=\cH_E\otimes\cH_B$ and let
$L=\Pi_E\otimes\Pi_B$ be a product projector.  Define
\begin{equation}
  \rho'=L\rho L,
  \qquad
  q_E=\tr((\rho'_E)^2),
  \qquad
  q_B=\tr((\rho'_B)^2).
  \label{eq:qE-qB}
\end{equation}

\begin{lemma}[Purity of a projected codeword]
\label{lem:projected-purity}
Let $|\phi\rangle\sim\Scr(\rho)$ and assume
$\lambda=\opnorm{\rho}<1$.  Then
\begin{align}
  M_L
  &:=\mathbb E\left[\tr\!\left(
    \tr_E(L\ketbra\phi L)^2
  \right)\right] =\mathbb E\left[\tr\!\left(
    \tr_B(L\ketbra\phi L)^2
  \right)\right]
  \le\frac{q_E+q_B}{1-\lambda}.
  \label{eq:projected-purity-bound}
\end{align}
\end{lemma}

\begin{proof}
For each projected vector $|x\rangle=L|\phi\rangle$, the two (unnormalized) reduced states $\tr_E(\ketbra x)$ and $\tr_B(\ketbra x)$ have the same spectrum (by the Schmidt decomposition).  Their purities are therefore
identical, proving the equality in Eq.~\eqref{eq:projected-purity-bound}.

To prove the upper bound, let $|x_a\rangle=L|a\rangle$ and
$X_a=\ketbra{x_a}$.  Contracting Eq.~\eqref{eq:exact-scrooge-second} with the partial swap on $B$
gives, after some rearrangement,
\begin{align}
  M_L
  =\sum_{a,b}p_ap_bK_{ab}\left(
  \tr((X_a)_B(X_b)_B) + \tr((X_a)_E(X_b)_E)\right).
  \label{eq:ML-expansion}
\end{align}
Since all $X_a$ are PSD, so too are their partial traces. Thus, the summands in parentheses are all nonnegative. Applying Eq.~\eqref{eq:Kab-bound} and summing over $a,b$ yields
\begin{equation}
  M_L\le\frac{1}{1-\lambda}
  \left(\tr((\rho'_B)^2)+\tr((\rho'_E)^2)\right).
\end{equation}
\end{proof}
\subsection{Scrooge Decoupling Inequality}
In this section, we prove a non-asymptotic, average decoupling inequality for Scrooge codes corresponding to a general density matrix $\rho$.

Introduce an auxiliary system $R$ and let $\ket{\omega_K}_{RQ} = \frac{1}{\sqrt{K}} \sum_{i=1}^K \ket{i}_R \ket{i}_Q$ denote the maximally entangled state on $RQ$, for $\cH_Q \cong \C^K$. Let $V: \cH_Q \to \cH$ be the Scrooge coding isometry in Eq.~\eqref{eq:lowdin}, and define
\begin{equation}
  |\Omega\rangle_{RS}
  =(\id_R\otimes V)|\omega_K\rangle.
\end{equation}
We decompose $\cH = \cH_E \otimes \cH_B$ for fixed subsystems $E$ and $B$. Let $\Omega = \ketbra{\Omega}$. Write the decoupling error as
\begin{equation}
  \Delta(V)
  := \Big\| \Omega_{RE}- \frac{\id_R}{K} \otimes \Omega_E \Big\|_1,
  \label{eq:Delta-definition}
\end{equation}
which we use throughout. Again, let $L = \Pi_E \otimes \Pi_B$ be a product projector. Further, let
\begin{equation}
  p_\eta=\mathrm{Pr}\{\opnorm{G-\id_K}>\eta\},
  \qquad
  \ell=1 - \tr(L \rho),
  \qquad
  d_E=\rank\Pi_E.
\end{equation}

\begin{theorem}[Scrooge decoupling inequality]
\label{thm:finite-scrooge-decoupling}
There is a universal constant $C>0$ such that, whenever
$0<\eta\le1/4$ and $\lambda=\opnorm{\rho}\le1/4$, 
\begin{align}
  \mathbb E\Big\| \Omega_{RE}- \frac{\id_R}{K} \otimes \Omega_E \Big\|_1
  \le C\bigg(&
    \sqrt{K d_E q_B}
    +\sqrt{d_E(\lambda+\eta)q_E}
    +\sqrt{K d_Ep_\eta}+\sqrt{\frac{\ell}{1-\eta}+p_\eta}
  ~\bigg),
  \label{eq:finite-decoupling}
\end{align}
for any product projector $L = \Pi_E \otimes \Pi_B$. Here, $q_E$ and $q_B$ are defined in Eq.~\eqref{eq:qE-qB}.
\end{theorem}

\begin{proof}
Let $\Omega'=(\id_R \otimes L)\Omega(\id_R \otimes L)$. By triangle inequality and contractivity of trace norm under partial trace, we have
\begin{align}
    \Delta(V) &\le \Big\| \Omega'_{RE} - \frac{\id_R}{K} \otimes \Omega'_E \Big\|_1 + \| \Omega_{RE} - \Omega'_{RE} \|_1 + \| \Omega_{E} - \Omega'_{E} \|_1 
    \\
    &\le \Big\| \Omega'_{RE} - \frac{\id_R}{K} \otimes \Omega'_E \Big\|_1 + 2 \| \Omega_{RE} - \Omega'_{RE} \|_1.
\end{align}
Let $P'=LPL$ and $\widetilde P'=L\widetilde PL$, for $P$ and $\widetilde{P}$ the normalized and unorthogonalized operator in Eq.~\eqref{eq:P-and-Ptilde}. We denote their reduced
operators by $P'_E:=\tr_B P'$, $P'_B:=\tr_E P'$, $\widetilde P'_E:=\tr_B\widetilde P'$, and $\widetilde P'_B:=\tr_E\widetilde P'$. Again using contractivity of trace norm, along with the gentle measurement lemma~\cite{winter_coding_1999}, we have $\| \Omega_{RE} - \Omega'_{RE} \|_1 \le \| \Omega - \Omega' \|_1 \le 3 \sqrt{1 - \tr((\id_R \otimes L)\Omega)} = 3 \sqrt{\delta_L}$ for $\delta_L = 1 - \tr(LP)/K$; here we have used $\Omega_S = P/K$. Thus, by Jensen's inequality for the square root
\begin{align}
    \E[\Delta(V)] \le \E\Big\| \Omega'_{RE} - \frac{\id_R}{K} \otimes \Omega'_E \Big\|_1 + 6 \sqrt{\E[\delta_L]}.
\end{align}
We bound $\E[\delta_L]$ by conditioning on the good and bad Gram events. On the bad Gram event, we simply use $\delta_L \le 1$. On the good Gram event, Eq.~\eqref{eq:loewner-comparison} gives
$P\preceq(1-\eta)^{-1}\widetilde P$, so $\delta_L \le \mathrm{tr}((\id-L) \widetilde{P})/(K (1-\eta))$. Using Eq.~\eqref{eq:scrooge-first-moment}, we obtain
\begin{equation}
  \mathbb E[\delta_L]
  \le\frac{\ell}{1-\eta}+p_\eta.
  \label{eq:leakage-bound}
\end{equation}
We bound the first term by switching to Frobenius norm; after some algebra, we have the exact identity
\begin{equation}
  \Big\|\Omega'_{RE}
  -\frac{\id_R}{K}\otimes\Omega'_E\Big\|_2^2
  =\frac{1}{K^2}\tr((P'_B)^2)
  -\frac{1}{K^3}\tr((P'_E)^2).
  \label{eq:H2-code-projector}
\end{equation}
Indeed, the first term is $\tr((\Omega')_{RE}^2)$, while the second is 
$\tr((\id_R/K \otimes \Omega'_E)^2)$. On the good Gram event,
$(1+\eta)^{-1}\widetilde P'\preceq P'\preceq (1-\eta)^{-1}\widetilde P'$, while on the bad Gram event we simply use that the traces of all of $(P'_E)^2, (P'_B)^2, (\widetilde{P}'_E)^2, (\widetilde{P}'_B)^2$ are at most $K^2$. Consequently,
\begin{align}
  \E \Big\|\Omega'_{RE}
  -\frac{\id_R}{K}\otimes\Omega'_E\Big\|_2^2
  \le \frac{(1-\eta)^{-2}}{K^2}\E[\mathrm{tr}((\widetilde{P}'_B)^2)]
  -\frac{(1+\eta)^{-2}}{K^3}\E[\mathrm{tr}((\widetilde{P}'_E)^2)] + 2 p_\eta.
  \label{eq:H2-loewner}
\end{align}
We have
\begin{align}
    \E[\mathrm{tr}((\widetilde{P}'_B)^2)] &= \E\Big[\tr\Big(\tr_E\Big(\sum_{i=1}^K L \dyad{\phi_i}L\Big)^2\Big)\Big]
    \\
    &= K(K-1) \tr(\tr_E(L\rho L)^2) + K \E[\tr(\tr_E(L \dyad{\phi} L)^2)] \label{eq: independence}
    \\
    &= K(K-1)q_B + K M_L
\end{align}
where in Eq.~\eqref{eq: independence} we have used independence of the $\ket{\phi_i}$ and the Scrooge first-moment condition. Similarly, we have $\E[\mathrm{tr}((\widetilde{P}'_E)^2)] = K(K-1) q_E + K M_L$. After some algebra, we find that
\begin{align}
     \E \Big\|\Omega'_{RE}
  -\frac{\id_R}{K}\otimes\Omega'_E\Big\|_2^2
  \le (1-\eta)^{-2} \frac{2 - \lambda}{1-\lambda} q_B + \frac{1}{K} \left( \frac{(1-\eta)^{-2}}{1-\lambda} - (1 + \eta)^{-2} \right) q_E + 2 p_\eta .
\end{align}
Finally, using $\lambda, \eta \le 1/4$, we have
\begin{align}
    \E \Big\|\Omega'_{RE}
  -\frac{\id_R}{K}\otimes\Omega'_E\Big\|_2^2
  \le 5 \left( q_B + \frac{\lambda + \eta}{K} q_E \right) + 2 p_\eta. 
\end{align}
$\Omega_{RE}'$ and $(\id_R/K) \otimes \Omega_E'$ are both supported on a space of dimension at most $K d_E$, so by Cauchy-Schwarz we have $\E\| \Omega'_{RE} - (\id_R / K) \otimes \Omega'_E \|_1 \le \sqrt{K d_E} (\E \| \Omega'_{RE} - (\id_R / K) \otimes \Omega'_E \|_2^2)^{1/2}$. Applying $\sqrt{x+y} \le \sqrt{x} + \sqrt{y}$ for $x,y \ge 0$, the result follows.
\end{proof}

\subsection{Scrooge Rate for Entropy-concentrated Spectrum}

Having proven a decoupling inequality for Scrooge codes associated with a general state $\rho$, we now specialize to sequences of states $\rho_N$ for which most of the negative log-eigenvalues concentrate around their expected value $S(\rho_N)$. For a sequence with entropy-concentration parameters (defined below) $w_N, \epsilon_N \rightarrow 0$, all but vanishing spectral weight lies on a subspace with dimension $2^{S(\rho_N) + o(N)}$, within which the eigenvalues are $2^{-S(\rho) + o(N)}$. Thus, the state is spectrally flat at leading exponential order. For such states, we prove an asymptotic average decoupling inequality for the corresponding random Scrooge code, which in turn gives a lower bound on the achievable rate of that Scrooge code. In this section and what follows, we assume for simplicity that values like $r N$, $p N$, etc., which define numbers of qubits, are integers. We could generally use floor functions to treat the general case; this does not change our asymptotic conclusions. 

\begin{defn}[Entropy-concentrated spectrum]
\label{def:entropy-concentrated}
Let $\tau=\sum_j\lambda_j\ketbra j$ be a state on $m$ qubits.  For $w\ge0$, define
\begin{equation}
  \Pi_\tau(w)
  =\sum_{j:\,| -\log_2 \lambda_j-S(\tau)|\le mw}\ketbra j.
  \label{eq:entropy-window}
\end{equation}
We say that $\tau$ has a $(w,\epsilon)$-entropy-concentrated spectrum if
\begin{equation}
  \tr(\Pi_\tau(w)\tau)\ge1-\epsilon.
\end{equation}
\end{defn}
We note that the entropy-concentration projector $\Pi_\tau(w)$ satisfies 
\begin{equation}
  \rank\Pi_\tau(w)\le2^{S(\tau)+mw},
  \qquad
  \Pi_\tau(w)\tau\Pi_\tau(w)
  \preceq2^{-S(\tau)+mw}\Pi_\tau(w).
  \label{eq:window-properties}
\end{equation}
To derive these properties, note that, by construction, $2^{-S(\tau) - mw}\le \lambda_j \le 2^{-S(\tau) + mw}$ for any $j$ such that $\ket{j}$ is in the support of $\Pi_\tau(w)$. The first result in Eq.~\eqref{eq:window-properties} follows from the lower bound on $\lambda_j$, while the second result in Eq.~\eqref{eq:window-properties} follows from the upper bound.

We choose the projectors $\Pi_E=\Pi_{\rho_E}(w_E)$ and $\Pi_B=\Pi_{\rho_B}(w_B)$ and let
$L=\Pi_E\otimes\Pi_B$. We again let $\rho' = L \rho L$. We first observe that
\begin{align}
    d_E&\le2^{S(\rho_E)+|E|w_E},\label{eq:dE-window}
\end{align}
which follows from the first inequality in Eq.~\eqref{eq:window-properties}. We next claim that
\begin{align}
    q_B&\le2^{-S(\rho_B)+|B|w_B} , \quad q_E\le2^{-S(\rho_E)+|E|w_E} . \label{eq:qB-window}
\end{align}
To see this, we first observe the generic fact that $\rho_B' \preceq \Pi_B \rho_B \Pi_B$, with the same true for $B$ replaced by $E$. It suffices to show that $\tr_E(\Pi_E \rho \Pi_E) \preceq \tr_E(\rho)$, since conjugating by $\Pi_B$ preserves PSD order; for any $\ket{b}$ on $\cH_B$ we have $\bra{b} \tr_E(\Pi_E \rho \Pi_E) \ket{b} \le \tr(\rho ( \id_E \otimes \ketbra{b})) = \bra{b} \tr_E(\rho) \ket{b}$, so the result follows. Now, applying the second result in Eq.~\eqref{eq:window-properties}, we may conclude that $(\rho_B')^2 \preceq 2^{-S(\rho_B)+|B| w_B} \rho_B'$; taking the trace of both sides and using $\tr(\rho_B') \le 1$ implies the bound on $q_B$, with the bound on $q_E$ following identically with $B$ replaced by $E$. Finally, we claim that
\begin{align}
  \ell&\le\epsilon_E+\epsilon_B.\label{eq:ell-window}
\end{align}
To see this, observe that $\id_S - \Pi_E \otimes \Pi_B \preceq (\id_E - \Pi_E) \otimes \id_B + \id_E  \otimes (\id_B - \Pi_B)$. Multiplying both sides by $\rho$ and taking the trace gives $\ell \le \tr((\id_E - \Pi_E) \rho_E) + \tr((\id_B - \Pi_B) \rho_B) \le \epsilon_E + \epsilon_B$ by definition. We write $S_\infty(\rho):= - \log_2 \|\rho\|_{\mathrm{op}}$ for the min-entropy of $\rho$. 

We now use these results to bound the decoupling error for the Scrooge random code corresponding to a state with entropy-concentrated spectrum:

\begin{theorem}[Scrooge code rate for entropy-concentrated states]\label{thm:general-scrooge}
Let $\rho_N$ be a family of $N$-qubit states on a system $[N]=E_N \cup B_N$, where $|E_N|/N\to p$. Suppose that $\rho_{E_N}$ and
$\rho_{B_N}$ have $(w_{E_N},\epsilon_{E_N})$- and
$(w_{B_N},\epsilon_{B_N})$-entropy-concentrated spectra, where $w_{E_N},w_{B_N}, \epsilon_{E_N}, \epsilon_{B_N} \to 0$. Further, suppose that $S(\rho_{E_N})/|E_N| \to s_E$, $S(\rho_{B_N})/|B_N| \to s_B$, and $S_\infty(\rho_N)/N \to s_\infty$. Let $K_N = 2^{rN}$ for $r \ge 0$; so long as
\begin{align}
    r < \min \{s_\infty, (1-p) s_B - p s_E\}
\end{align}
we have $\E[\Delta(V_N)] \to 0$ for $V_N$ the coding isometry of the Scrooge random code corresponding to $\rho_N$.
\end{theorem}
\noindent
By Markov's inequality, a Scrooge random code satisfying the above conditions has error decaying to zero with high probability.

As an immediate corollary, if $\rho_N$ are uniform volume-law states, so that $S(\rho_{E_N})/|E_N|$ and $S(\rho_{B_N})/|B_N|$ both converge to $s$, then every rate
\begin{align}
    r < \min\{s_\infty, s(1-2p)\}
\end{align}
is achievable with high probability against the specified sequence of erasures $E_N$. 

\begin{proof}
We first note that $r < s_\infty$ implies that, for large $N$, $K_N \le \mathrm{rank}(\rho_N)$, so the Scrooge random vectors are linearly independent almost surely. We next establish some minor notation. By assumption, there exist constants $a, \gamma > 0$ such that $2a + \gamma < s_\infty - r$ (constant because $r$ and $s_\infty$ are independent of $N$). Let $\eta_N = 2^{-aN}$ be the bound on $\|G - \id_K\|$ in the good Gram event.

We bound each of the terms in Eq.~\eqref{eq:finite-decoupling} using the results of Eq.~\eqref{eq:dE-window}, Eq.~\eqref{eq:qB-window}, and Eq.~\eqref{eq:ell-window}. For the first term, we have
\begin{align}
    K_N d_{E_N} q_{B_N} &\le 2^{rN} 2^{S(\rho_{E_N}) + |E_N| w_{E_N}}2^{-S(\rho_{B_N}) + |B_N| w_{B_N}}  
    \\
    &\le 2^{(r + s_E p + p w_{E_N} - s_B(1-p) + (1-p) w_{B_N} + o(1))N}.
\end{align}
The exponent is $r + s_E p + p w_{E_N} - s_B(1-p) + (1-p) w_{B_N} \to r - ((1-p) s_B - p s_E) < 0$, so this term decays exponentially to zero. Set $\lambda_N:= \|\rho_N\|_\mathrm{op} = 2^{-s_\infty N+o(N)}$. For the second term, we have
\begin{align}
    d_{E_N} (\lambda_N + \eta_N) q_{E_N} &\le 2^{2 |E_N| w_{E_N}} (\lambda_N + \eta_N) 
    \\
    &\le 2^{(2p w_{E_N} - s_\infty + o(1))N} + 2^{(2p w_{E_N} - a + o(1))N}.
\end{align}
The exponents are $2p w_{E_N} - s_\infty + o(1) \to -s_\infty < 0$ and $2p w_{E_N} - a + o(1) \to -a < 0$, so this term also decays exponentially to zero. For the third term, Lemma~\ref{lem:gram} gives, 
\begin{equation}
    p_{\eta_N} \le C \sqrt{K_N} \exp\left(\frac{K_N \lambda_N}{2} - c \frac{\eta_N^2}{K_N \lambda_N}\right)
\end{equation}
Since $K_N \lambda_N = 2^{-(s_\infty -r)N + o(N)}$ and $\eta_N = 2^{-aN}$, we have, \begin{equation}
    \frac{\eta_N^2}{K_N\lambda_N} = 2^{(s_\infty -r-2a)N + o(N)}.
\end{equation}
Because $s_\infty -r-2a >\gamma >0$, for all sufficiently large $N$, 
\begin{equation}
    p_{\eta_N} \le \exp(-c'2^{\gamma N})
\end{equation}
for some constant $c'>0$ after absorbing the exponentially growing prefactor $\sqrt{K_N}$ and the vanishing $K_N \lambda_N/2$ term. Since $K_Nd_{E_N}$ grows at most exponentially in $N$, it follows that $K_Nd_{E_N}p_{\eta_N} \rightarrow 0$ doubly exponentially. Therefore, this term decays doubly exponentially to zero. Finally, for the fourth term, the $p_{\eta_N}$ contribution again decays doubly exponentially to zero, while 
\begin{align}
    \frac{\ell_N}{1 - \eta_N} \le O(\ell_N) \le O(\epsilon_{E_N} + \epsilon_{B_N}) \rightarrow 0.
\end{align}
Hence the corresponding square root term also vanishes, concluding the proof.
\end{proof}

\section{Saturation of the Entropic Quantum Singleton Bound for all Parameters}
\label{sec:tightness}

In Sec.~\ref{sec:scrooge_achieve}, we derived achievable rates for Scrooge random codes with general entropy-concentrated parent states. We now construct explicit examples of parent states which are entropy-concentrated and saturate the entropic quantum Singleton bound. Then, we specialize to Gibbs states, obtaining the corresponding finite-temperature result. Combined with the converse of Sec.~\ref{sec: entropic_singleton}, these constructions establish saturation of the entropic Singleton bound — unconditionally for the bounded-entropy states, and under Assumption~\ref{assump:gibbs-spectral-aep} for the thermal case.

\subsection{Bounded Entropy Scrooge Code Construction}

We define the states $\rho_N$ which we use to construct the saturating Scrooge code. Fix an entropy density $s \in (0,1]$. Let $q\in(0,1/2]$ be such that $H_2(q) = s$, for $H_2(\cdot)$ the binary entropy function. Let $m_N= qN $ and define the constant-weight subspace  (type class)
\begin{equation}
  \cT_{m_N}
  =\mathrm{span}\{\,|x\rangle:x\in\{0,1\}^N,\ |x|=m_N\,\}.
  \label{eq:type-subspace}
\end{equation}
Let $d_N = \dim(\cT_{m_N})=\binom{N}{m_N}$ and let $\Pi_N$ denote the projector onto $\cT_{m_N}$. We define $\rho_N$ as the maximally mixed state over $\cT_{m_N}$,
\begin{equation}
  \rho_N=\frac{\Pi_N}{d_N}.
  \label{eq:type-state}
\end{equation}
We call these states the \emph{type states}. By Stirling's formula, we have $\log(d_N) = sN + O(\log (N))$, so $\rho_N$ indeed satisfies the desired entropy bound up to subleading factors. 
Because $\rho_N$ is maximally mixed on its support, its Scrooge ensemble is exactly the Haar
measure on that support:
\begin{equation}
  \Scr(\rho_N)=\mathrm{Haar}(\cT_{m_N}).
  \label{eq:flat-scrooge-haar}
\end{equation}
Indeed, in Definition~\ref{def:scrooge}, both the reweighting density and the normalization of
$\sqrt{\rho_N}|\psi\rangle$ are constant.  Consequently, the span of $K\le d_N$
independent Scrooge samples is a Haar-random $K$-dimensional subspace of
$\cT_{m_N}$.

We now prove that the reduced density matrices of the type states have entropy-concentrated spectra:

\begin{lemma}[Type state has entropy-concentrated spectrum]\label{lem:type-local-spectrum}
Fix $0<\alpha<1$ and let $A\subset[N]$ satisfy $|A|=\alpha N+O(1)$.  Uniformly over the
choice of $A$,
\begin{equation}
  S((\rho_N)_A)=s|A|+O(N^{3/4}),
  \label{eq:type-entropy}
\end{equation}
and $(\rho_N)_A$ has a
$(CN^{-1/4},2^{-c\sqrt N})$-entropy-concentrated spectrum for constants $c,C>0$
depending only on $q$ and $\alpha$.
\end{lemma}

\begin{proof}
Write $a_N=|A|$.  The state $(\rho_N)_A$ is diagonal in the computational
basis.  Every string of weight $j$ has eigenvalue
\begin{equation}
  \lambda_j=\frac{\binom{N-a_N}{m_N-j}}{\binom{N}{m_N}},
  \label{eq:type-eigenvalue}
\end{equation}
and there are $\binom{a_N}{j}$ such strings. Therefore, if a string is sampled according to the spectrum of $\rho_N$, the weight of the string $J$ is hypergeometric:
\begin{equation}
  \mathrm{Pr}\{J=j\}
  =\frac{\binom{a_N}{j}\binom{N-a_N}{m_N-j}}{\binom{N}{m_N}}.
  \label{eq:hypergeometric}
\end{equation}
Because sampling without replacement concentrates faster than sampling with replacement, Hoeffding's inequality~\cite{hoeffding_probability_1963} implies
$\mathrm{Pr}\{|J-q a_N|>N^{3/4}\}\le2^{-c\sqrt N}$ for some $c > 0$.

On the event $|J-qa_N|\le N^{3/4}$, we have
$(m_N-J)/(N-a_N)=q+O(N^{-1/4})$.  By Stirling's inequality and the fact that $H_2$ is smooth near any $q \in (0,1/2]$, we have
\begin{align}
  -\log\lambda_J &= N H_2(q) - (N-a_N) H_2(q + O(N^{-1/4}))  + O(\log(N)) 
    \\
    &= a_N H_2(q) + O(N^{3/4})
    \\
    &= s a_N + O(N^{3/4})
  \label{eq:type-surprisal}
\end{align}
Every nonzero eigenvalue in Eq.~\eqref{eq:type-eigenvalue} is at least $d_N^{-1}$, so
$0\le-\log\lambda_J\le\log d_N=O(N)$.  Since
$S((\rho_N)_A)=\mathbb E[-\log\lambda_J]$, the exponentially small tail and Eq.~\eqref{eq:type-surprisal} prove Eq.~\eqref{eq:type-entropy}.  Combining the same two estimates shows
that, except with probability $2^{-c\sqrt N}$,
$|{-\log\lambda_J-S((\rho_N)_A)}|\le C N^{3/4}$.  Because $a_N=\alpha N+O(1)$,
this is precisely entropy concentration with width $CN^{-1/4}$. 
\end{proof}

\subsection{Bounded Entropy Tight Rate}

We can now combine Theorem~\ref{thm:general-scrooge} with Lemma~\ref{lem:type-local-spectrum} to conclude a bound on the achievable rate of the Scrooge random code for the maximally mixed type state. In fact, we go further, proving that the code corrects \emph{all} erasures up to fixed size with all but vanishing probability, thus proving the achievability of the entropic quantum Singleton bound. 

To show that all erasures up to fixed size can be corrected, we use a standard concentration inequality for functions of Haar random unitaries~\cite{meckes_spectral_2013}. Specifically, if $f: \mathrm{U}(d) \to \R$ is $L$-Lipschitz in Frobenius norm, then for Haar random $U$
\begin{equation}
  \mathrm{Pr}\{f(U)>\mathbb E[f(U)]+t\}
  \le \exp\!\left(-c d t^2/L^2\right)
  \label{eq:unitary-concentration}
\end{equation}
for every $t > 0$, for some universal constant $c > 0$.

\begin{theorem}[Saturation of the entropic quantum Singleton bound]
\label{thm:full-tightness}
Fix $s\in(0,1]$ and $p\in[0,1/2)$.  Let $\rho_N$ be the state in
 Eq.~\eqref{eq:type-state}, let $e_N = pN$, and let $K_N= 2^{s (1-2p) N-N/\log N}.$
If $V_N$ is the coding isometry associated to $K_N$ independent samples from
$\Scr(\rho_N)$, then, with all but doubly exponentially small probability,
\begin{align}
    \min_{\substack{A\cap C=\varnothing\\|A|=|C|=e_N}}
  S\!\left((\overline\rho_N)_{[N]\setminus(A\cup C)}\right)
  =s (1-2p) N+o(N),\label{eq:tight-entropy}
\end{align}
where $\overline\rho_N=V_N V_N^\dagger/K_N$, while also $\max_{|E|\le e_N}\Delta(V_N)\to 0$ stretched-exponentially, for $\Delta(V_N)$ the decoupling error defined in Eq.~\eqref{eq:Delta-definition}.
Thus, with high probability, these codes saturate the entropic quantum Singleton bound with rate 
\begin{align}
    r = s (1-2p).
\end{align}
\end{theorem}

\begin{proof}
Assume $p > 0$; the case $p = 0$ is trivial.
Fix an erased subregion $E_N$ of size $e_N = pN$ and let $B_N$ denote its complement. By Lemma~\ref{lem:type-local-spectrum}, both $\rho_{E_N}$ and $\rho_{B_N}$ have $(CN^{-1/4},2^{-c\sqrt N})$-entropy-concentrated spectra, and satisfy
\begin{align}
    \frac{S(\rho_{E_N})}{|E_N|}, \frac{S(\rho_{B_N})}{|B_N|} \to s.
\end{align}
Moreover, we have $S_\infty(\rho_N) = \log \binom{N}{m_N} = sN + O(\log(N))$, so $S_\infty(\rho_N)/N \to s$. 
We now carefully bound the expected decoupling error $\E[\Delta(V_N)]$. We slightly modify the construction in the proof of Theorem~\ref{thm:general-scrooge} (to take advantage of the fact that $s_\infty = s_E = s_B$) by choosing $\eta_N = 2^{-s p N /2}$. The same analysis as in that proof shows
\begin{align}
  \mathbb E[\Delta(V_N)] &\le O\left(2^{-\frac{N}{2\log N}+O(N^{3/4})}+2^{-s p N/4+O(N^{3/4})}+ 2^{\left( \frac{3r}{4} + \frac{sp}{2} \right)N} 2^{\frac{p}{2} 2^{-c\sqrt{N}}N} 2^{-c' 2^{spN + \frac{N}{\log N} - O(\log N)}}+2^{-c\sqrt N}\right)  \label{eq:type-fixed-erasure-mean}
\end{align}
which is at most $O(2^{-c\sqrt N})$, which indeed decays stretched-exponentially to zero.

Now we must show that, with high probability, $\Delta(V_N)$ decays stretched-exponentially to zero simultaneously for all erased subregions $E$ of size $e_N$. 
$\Delta(V_N)$ depends only on the codespace corresponding to $V_N$, which is a Haar random subspace of $\cT_{m_N}$ of dimension $K_N$. Therefore, we can replace $V_N$ by $V(U) = U V_{0}$ for $U$ Haar random on $\cT_{m_N}$ and $V_0: \C^{K_N} \to \cT_{m_N}$ a fixed isometry, without changing the expected value. Let $f(U) = \Delta(V(U))$; we now compute the Lipschitz constant of $f$. Let $\ket{\Omega(U)} = \frac{1}{\sqrt{K_N}} \sum_{i=1}^{K_N} \ket{i}_R \otimes U V_0 \ket{i}$ and $\Omega(U) = \ketbra{\Omega(U)}$, and similarly for $\ket{\Omega(W)} $.
We have
\begin{align}
    |f(U) - f(W)| &\le \| \Omega_{RE}(U) - \Omega_{RE}(W) \|_1 +  \| \Omega_{E}(U) - \Omega_{E}(W) \|_1 
    \\
    &\le 2 \| \Omega(U) - \Omega(W) \|_1
    \\
    &\le 4 \| \ket{\Omega(U)} - \ket{\Omega(W)} \|_2,
\end{align}
where the first line is triangle inequality, the second is contractivity of trace norm under partial trace, and the third is the relationship between trace distance and Euclidean distance for pure states. We then have
\begin{align}
    |f(U) - f(W)| &\le 4 \Big( \frac{1}{K_N} \sum_{i=1}^{K_N} \| (U-W) V_0 \ket{i} \|_2^2 \Big)^{1/2}
    \\
    &= 4 \Big( \frac{1}{K_N} \| (U-W) V_0 \|_2^2 \Big)^{1/2}
    \\
    &\le \frac{4}{\sqrt{K_N}} \| U-W\|_2,
\end{align}
so $f$ is $4/\sqrt{K_N}$-Lipschitz. Applying the concentration inequality Eq.~\eqref{eq:unitary-concentration} and taking a union bound over all erased subregions of size $e_N$ (of which there are, generously, at most $2^N$), we have
\begin{align}
    \mathrm{Pr}\left\{\underset{|E|=e_N}{\max} \Delta(V_N) > \E[\Delta(V_N)] + t\right\} \le 2^N 2^{-c' d_N K_N t^2}
\end{align}
for some constant $c' > 0$. We have $\log(d_N K_N) = 2s(1-p) N - N/\log N + O(\log N) = \Omega(N)$, so this is
\begin{align}
    \mathrm{Pr}\left\{\underset{|E|=e_N}{\max} \Delta(V_N) > \E[\Delta(V_N)] + t\right\} \le 2^{N - c' 2^{c'' N} t^2}
\end{align}
for some constant $c'' > 0$. Choose $t = 2^{-c\sqrt{N}}$. Then, with all but doubly exponentially small probability, the largest decoupling error for any subset of size $e_N$ is at most $O(2^{-c\sqrt{N}})$. The decoupling error for any subset of size less than $e_N$ is automatically smaller, by contractivity of trace norm under partial trace. This proves simultaneous correctability of all possible erasures of at most $p N$ qubits.

Finally, we prove Eq.~\eqref{eq:tight-entropy} to demonstrate that the code satisfies the conditions for the entropic quantum Singleton bound, thus establishing its saturation. On the high-probability event of simultaneous decoupling established above, by the pairwise entropic Singleton bound of Theorem~\ref{thm:entropic_singleton}, together with its asymptotic approximate extension established in Corollary~\ref{corr:entropic_singleton_volume_law}, implies, uniformly over all disjoint $A$, $C$ of size $e_N$, \begin{equation}
    S((\overline{\rho}_N)_{[N] \backslash (A \cup C)}) \geq \log_2 K_N - o(N) = s(1-2p)N - o(N).
\end{equation}
It therefore remains only to prove the matching upper bound. Fix a subsystem $D$ of size $(1-2p)N$, with complement $F$ of size $2pN$. It suffices to show that $S((\overline{\rho}_N)_D) \le s(1-2p)N + o(N)$. 
We begin by rewriting $(\overline{\rho}_N)_D$ in a convenient form. For each $j \in \{0, \dots, m_N\}$, let $\Pi_j^D$ denote the projector onto the weight $j$ subspace of $\cH_D$ and let $\Pi_{m_N-j}^F$ denote the projector onto the weight $m_N-j$ subspace of $\cH_F$. We have, since $\overline{\rho}_N$ is entirely supported in $\cT_{m_N}$,
\begin{align}
  (\overline{\rho}_N)_D = \bigoplus_{j} p_j \sigma_j
\end{align} 
for $p_j = \mathrm{tr}(\overline{\rho}_N (\Pi_j^D \otimes \Pi_{m_N - j}^F)) = \mathrm{tr}((\overline{\rho}_N)_D \Pi_j^D)$ and $\sigma_j$ a density matrix supported entirely on the weight $j$ subspace of $D$, $\mathrm{im}(\Pi_j^D)$. Since $(\overline{\rho}_N)_D$ is block-diagonal, its entropy is
\begin{align}\label{eq: entropy_formula}
  S((\overline{\rho}_N)_D) &= H(\{p_j\}) + \sum_j p_j S(\sigma_j) \le \log(N+1) + \sum_j p_j \log \binom{(1-2p)N}{j},
\end{align}
where $H(\cdot)$ denotes the Shannon entropy~\cite{shannon1948mathematical}, and the inequality follows from the fact that $\sigma_j$ is supported on $\Pi_j^D$. We claim that $\{p_j\}$ is close to a hypergeometric distribution with high probability. We can again show this using the concentration inequality Eq.~\eqref{eq:unitary-concentration}. We first claim that $\E[p_j]$ is hypergeometric. Indeed, $\overline{\rho}_N$ is distributed as $U P_0 U^\dagger / K_N$ for $P_0$ the projector onto some fixed $K_N$-dimensional subspace of $\cT_{m_N}$. Hence,
\begin{align}
    \E[p_j] = \tr(\E\left[\frac{U P_0 U^\dagger}{K_N}\right] (\Pi_j^D \otimes \Pi_{m_N - j}^F)) = \frac{1}{d_N} \tr(\Pi_j^D \otimes \Pi_{m_N - j}^F) = \frac{\binom{(1-2p)N}{j} \binom{2pN}{m_N - j}}{\binom{N}{m_N}},
\end{align}
using unitary invariance. This is indeed the hypergeometric distribution. We now compute the Lipschitz constant of $p_j(U)$. Letting $M_j$ denote $\Pi_j^D \otimes \Pi_{m_N - j}^F$ for simplicity, we have
\begin{align}
  |p_j(U) - p_j(W)| &= \frac{1}{K_N} \Big| \mathrm{tr}(P_0 (U^\dagger M_j U - W^\dagger M_j W)) \Big| 
  \\
  &\le \frac{1}{K_N} \| P_0 \|_2 \cdot \| U^\dagger M_j U - W^\dagger M_j W \|_2
  \\
  &= \frac{1}{\sqrt{K_N}} \| U^\dagger M_j U - W^\dagger M_j W \|_2,
\end{align}
using H\"{o}lder's inequality and the fact that $P_0$ is a rank $K_N$ projector. Using $\| U^\dagger M_j U - W^\dagger M_j W \|_2 \le 2 \| U - W \|_2$, by triangle inequality and H\"{o}lder's inequality, we conclude that $p_j$ is $2/\sqrt{K_N}$-Lipschitz. 
Applying Eq.~\eqref{eq:unitary-concentration} yields $\mathrm{Pr}\{|p_j - \E[p_j]| > t\} \le 2^{-c' d_N K_N t^2}$ for some constant $c' > 0$. 
This is true for all $N+1$ choices of $j$. Choosing $t = 2^{-c \sqrt{N}}$ for some constant $c > 0$, we have, as in the analysis above, that with all but doubly exponentially small probability, $|p_j - \E[p_j]| \le 2^{-c \sqrt{N}}$ for all $j$.

We condition on this event moving forward. Let $\cJ = \{j: |j - q (1-2p) N| \le N^{3/4}\}$ denote the set of weights close to the hypergeometric mean, and let $\cJ^c$ be its complement. We have $\sum_{j \in \cJ^c} \E[p_j] \le 2^{-c'' \sqrt{N}}$ for some $c'' > 0$, by Hoeffding's inequality. Thus,
\begin{align}
  \sum_{j \in \cJ^c} p_j \le \sum_{j \in \cJ^c} \E[p_j] + \sum_{j \in \cJ^c}| p_j - \E[p_j] | \le 2^{-c''' \sqrt{N}}
\end{align}
for some $c''' > 0$. Finally, we bound $S((\overline{\rho}_N)_D)$. Returning to Eq.~\eqref{eq: entropy_formula}, we have
\begin{align}
  S((\overline{\rho}_N)_D) &\le \log(N+1) + \sum_{j \in \cJ^c} p_j \log \binom{(1-2p)N}{j} + \sum_{j \in \cJ} p_j \log \binom{(1-2p)N}{j} 
  \\
  &\le \log(N+1) + (1-2p) N 2^{-c''' \sqrt{N}} + \sum_{j \in \cJ} p_j \log \binom{(1-2p)N}{j} .
\end{align}
If $j \in \cJ$ (i.e., $|j - q (1-2p) N| \le N^{3/4}$), then $\log \binom{(1-2p)N}{j} = s (1-2p) N + O(N^{3/4})$. Therefore,
\begin{align}
  S((\overline{\rho}_N)_D) \le s(1-2p)N + O(N^{3/4}) = s(1-2p)N + o(N)
\end{align}
with all but doubly exponentially small probability, as desired.

\end{proof}

\subsection{Thermal Scrooge Codes}\label{sec:thermal_scrooge}
To provide a comparison to our thermalization codes, we study the properties of the $g_\beta$-Scrooge code. We therefore need the following assumption about reduced Gibbs states. 

\begin{assumption}[Entropy concentration for reduced Gibbs states]
\label{assump:gibbs-spectral-aep}
Fix inverse temperature $\beta$, an erasure fraction $p$, and define $\mathcal E_{p,N} :=\{E\subseteq[N]:|E|=\lfloor pN\rfloor\}$. For each $E\in\mathcal E_{p,N}$, let $B=E^c$ and define
\begin{equation}
  \rho_{E,N}:=\tr_B(g_{\beta,N}),
  \qquad
  \rho_{B,N}:=\tr_E(g_{\beta,N}),
\end{equation}
where $g_{\beta,N}=\frac{e^{-\beta H_N}}{Z_N}$. There exist entropy-density profiles
$s_E(\beta,p),s_B(\beta,p)>0$,
uniformly over $E\in\mathcal E_{p,N}$, such that
\begin{align}
  S(\rho_{E,N})
  =pN\,s_E(\beta,p)+o(N),\qquad
  S(\rho_{B,N})
  =(1-p)N\,s_B(\beta,p)+o(N),
\end{align}
Moreover, for every $E\in\mathcal E_{p,N}$ there exist spectral widths $w_{E,N},w_{B,N}\ge0$ and leakage probabilities $\epsilon_{E,N},\epsilon_{B,N}\ge0$ such that
$\rho_{E,N}$ has a
$(w_{E,N},\epsilon_{E,N})$-entropy-concentrated spectrum and
$\rho_{B,N}$ has a
$(w_{B,N},\epsilon_{B,N})$-entropy-concentrated spectrum in the sense
of Definition~\ref{def:entropy-concentrated}, with 
\begin{align}
  \max_{E\in\mathcal E_{p,N}}
  \max\{w_{E,N},w_{B,N}\} \longrightarrow 0, \qquad \max_{E\in\mathcal E_{p,N}}
  \max\{\epsilon_{E,N},\epsilon_{B,N}\}
  \longrightarrow 0.
\end{align}
\end{assumption}

We now provide intuition for Assumption~\ref{assump:gibbs-spectral-aep}. For the full Gibbs state, entropy concentration is precisely a statement about energy. Ref.~\cite{anshu_concentration_2016} gives concentration bounds for the energy distribution of local Hamiltonians in states with finite correlation length. The physical intuition is that, when correlations decay beyond a finite length scale, the total energy is built from many approximately independent local contributions, so its fluctuations are subextensive compared with the system size. Indeed, for the full Gibbs state, the eigenvalues $\lambda_j = \frac{e^{-\beta E_j}}{Z}$ satisfy $- \log \lambda_j = \beta E_j + \log Z$. Therefore, the fluctuations of $-\log \lambda_j$ are exactly proportional to the thermal energy fluctuations, $\mathrm{Var}_j (-\log \lambda_j) = \beta^2 \mathrm{Var}_j(E_j)$, where $j$ is sampled with probability $\lambda_j$. Moreover, $\mathbb E_j[-\log \lambda_j] = S(g_{\beta,N})$. Thus, concentration of the Gibbs state's eigenvalues around the entropy is equivalent to concentration of the thermal energy around its mean, as given in \cite{anshu_concentration_2016}.

For a reduced Gibbs state $\rho_{A,N}:=\tr_{A^c}(g_{\beta,N})$, the above concentration no longer holds exactly. Nevertheless, the physical principle of the locality of temperature \cite{Kliesch2014} suggests that, at sufficiently high temperature, the thermal reduced state on $A$ behaves approximately like a Gibbs state of the Hamiltonian terms living in or near $A$. Heuristically, $\mathrm{tr}_{A^c}(e^{-\beta H_N})/Z_N
\approx
e^{-\beta H_A^{\text{eff}}}/\mathrm{tr}(e^{-\beta H_A^{\text{eff}}})$ where $H_A$ denotes the part of the Hamiltonian supported on $A$ with support contained entirely in $A$. If this local Gibbs picture is approximately valid, then one expects the reduced Gibbs state to have most of its spectral weight on eigenvalues whose whose negative logarithms are close to the entropy $S(\rho_{A,N})$. 

This expectation is also supported, at a qualitative level, by the
quantum Shannon--McMillan theorem \cite{bjelakovic_shannon-mcmillan_2002}. For translation-invariant ergodic
quantum lattice states, the restriction to a sufficiently large
regular region has an asymptotically typical spectral subspace whose
dimension is governed by the von Neumann entropy density. Equivalently,
with probability tending to one under the spectrum of the reduced
state, the information content per site $-\log\lambda_j/|A|$
approaches the entropy density. Thus, a qualitative concentration of
the spectrum around $S(\rho_{A,N})$ is expected for thermodynamic
Gibbs states in an ergodic phase. Assumption~\ref{assump:gibbs-spectral-aep} is a direct formulation of this expectation. Under Assumption~\ref{assump:gibbs-spectral-aep}, we can then state the following $g_\beta$-Scrooge code corollary. To strengthen Theorem~\ref{thm:general-scrooge} from a fixed-erasure
statement to simultaneous decoupling over all
$E\in\mathcal E_{p,N}$, we use the following concentration lemma.

\begin{lemma}[Uniform concentration of Scrooge decoupling]
\label{lem:uniform-scrooge-concentration}
Let $\rho_N$ be a family of $N$-qubit states satisfying $S_\infty(\rho_N)/N\to s_\infty>0$, and let $V_N$ be a $\rho_N$-Scrooge code of dimension $K_N=2^{rN}$ with
$r<s_\infty$. Then, for any family $\mathcal F_N$ of at most $2^N$
erased regions and every fixed $\epsilon>0$,
\begin{equation}
\Pr_{\mathrm{Scr}}\!\left\{
\max_{E\in\mathcal F_N}\Delta_E(V_N)
>
\max_{E\in\mathcal F_N}
\mathbb E_{\mathrm{Scr}}\Delta_E(V_N)
+\epsilon
\right\}
\longrightarrow0.
\end{equation}
\end{lemma}
\begin{proof}
The proof follows the concentration argument of Theorem~\ref{thm:full-tightness}. There, the
parent state is maximally mixed on its support, so the Scrooge codespace is Haar-random. For general $\rho_N$, we instead use the adjusted-Gaussian representation.

Let $G=(|g_1\rangle,\ldots,|g_{K_N}\rangle)$ have independent standard complex Gaussian columns and set $Y=\sqrt{\rho_N}G$. Weighting the Gaussian law by $\prod_{i=1}^{K_N}\langle g_i|\rho_N|g_i\rangle$ produces independent Scrooge samples after column normalization. Since column normalization does not change the span, and $\Delta_E$ is invariant under a change of logical basis, the encoding isometry may be represented by the polar factor $V(G)=Y(Y^\dagger Y)^{-1/2}$. 
Define $\lambda_N:=\|\rho_N\|_{\mathrm{op}}
=2^{-s_\infty N+o(N)}$. Since $r < s_\infty$, \begin{equation}
    K_N\lambda_N
=
2^{-(s_\infty-r)N+o(N)}
\longrightarrow0.
\label{eq:Klambda-small}
\end{equation}
For any fixed $0<\eta<1/2$, the Gaussian matrix estimate used in
Lemma~\ref{lem:gram} gives
\begin{equation}
    \Pr_0\!\left\{
\|Y^\dagger Y-\mathbbm 1\|_{\mathrm{op}}>\eta
\right\}
\le
C K_N
\exp\!\left(
-\frac{c}{K_N\lambda_N}
\right).
\label{eq:gaussian-good-gram}
\end{equation}
On the complementary good-Gram event, the same Lipschitz estimate used in Theorem~\ref{thm:full-tightness} gives, for any two coding isometries $V,W$,
\begin{equation}
    |\Delta_E(V)-\Delta_E(W)|
\le
\frac{4}{\sqrt{K_N}}\|V-W\|_2.
\end{equation}
The standard perturbation bound for polar factors, together with $\|\sqrt{\rho_N}(G-H)\|_2 \le \sqrt{\lambda_N}\|G-H\|_2$ therefore implies, uniformly over $E$,
\begin{equation}
|\Delta_E(V(G))-\Delta_E(V(H))|
\le
C\sqrt{\frac{\lambda_N}{K_N}}\,
\|G-H\|_2 .
\label{eq:general-scrooge-lipschitz}
\end{equation}
Extending from the good-Gram set to the full Gaussian space and applying Gaussian concentration yields
\begin{equation}
\Pr_0\!\left\{
\Delta_E(V_N)>
\mathbb E_0\Delta_E(V_N)+t
\right\}
\le
C\exp\!\left(
-c\frac{K_Nt^2}{\lambda_N}
\right)
+
C K_N\exp\!\left(
-\frac{c}{K_N\lambda_N}
\right),
\label{eq:general-scrooge-gaussian-tail}
\end{equation}
where $\Pr_0$ denotes the unweighted Gaussian law. It remains only to transfer this estimate to the Scrooge measure. Its likelihood ratio relative to the ordinary Gaussian law is $L(G)=
\prod_{i=1}^{K_N}\langle g_i|\rho_N|g_i\rangle$, and 
\begin{align}
\mathbb E_0[L^2]
&=
\left(1+\operatorname{tr}\rho_N^2\right)^{K_N}\le
e^{K_N\lambda_N}
=
1+o(1).
\label{eq:general-scrooge-likelihood}
\end{align}
Hence, by Cauchy--Schwarz, for any event $A$,
\begin{equation}
\Pr_{\mathrm{Scr}}(A)
\le
e^{K_N\lambda_N/2}\Pr_0(A)^{1/2}.
\label{eq:general-scrooge-measure-transfer}
\end{equation}
Moreover, since $0\le\Delta_E\le2$,
\begin{equation}
a_N:=
\max_{E\in\mathcal F_N}
\left|
\mathbb E_{\rm Scr}\Delta_E(V_N)
-
\mathbb E_0\Delta_E(V_N)
\right|
=o(1),
\end{equation}
where this follows from
$\mathbb E_0[L^2]=1+o(1)$ and hence convergence of the Scrooge and
Gaussian measures in total variation. Taking a union bound in
Eq.~\eqref{eq:general-scrooge-gaussian-tail} over
$|\mathcal F_N|\le2^N$ with $t=N^{-1}$ gives
\begin{align}
\Pr_0\!\left\{
\max_{E\in\mathcal F_N}\Delta_E(V_N)
>
\max_{E\in\mathcal F_N}\mathbb E_0\Delta_E(V_N)
+\frac1N
\right\}
&\le
C2^N
\exp\!\left(
-c\frac{K_N}{N^2\lambda_N}
\right)+
C K_N
\exp\!\left(
-\frac{c}{K_N\lambda_N}
\right)
\longrightarrow0.
\end{align}
Indeed,
$$
\frac{K_N}{\lambda_N}
=
2^{(r+s_\infty)N+o(N)},
\qquad
\frac{1}{K_N\lambda_N}
=
2^{(s_\infty-r)N+o(N)},
$$
so both tails are doubly exponentially small. Eq.~\eqref{eq:general-scrooge-measure-transfer} transfers the same
high-probability statement to the Scrooge measure. Therefore,
with probability tending to one,
\begin{align}
\max_{E\in\mathcal F_N}\Delta_E(V_N)
\le
\max_{E\in\mathcal F_N}
\mathbb E_{\rm Scr}\Delta_E(V_N)
+a_N+\frac1N = \max_{E\in\mathcal F_N}
\mathbb E_{\rm Scr}\Delta_E(V_N)
+o(1),
\end{align}
which proves the claim.

\end{proof}

\begin{corollary}[Thermal Scrooge codes] \label{corr:thermal-scrooge-uniform} Fix $\beta$ and $p\in[0,1/2)$, and suppose Assumption~\ref{assump:gibbs-spectral-aep} holds uniformly over all $E\subseteq[N]$ with $|E|=\lfloor pN\rfloor$, with $s_E(\beta,p)=s_B(\beta,p)=s(\beta)>0$. Let $V_N$ be a $g_{\beta,N}$-Scrooge code of rate $r$, where $S_\infty(g_{\beta,N})/N\to s_\infty(\beta)>0$. If \begin{equation} 0<r< \min\{s_\infty(\beta),s(\beta)(1-2p)\}, \end{equation} then, with probability tending to one, \begin{equation} \max_{|E|\leq\lfloor pN\rfloor} \left\| \Omega_{RE}(V_N) - \frac{\mathbbm 1_R}{K_N}\otimes\Omega_E(V_N) \right\|_1 \longrightarrow0. \label{eq:uniform-thermal-scrooge-decoupling} \end{equation} Moreover, there exists a single subcode $\mathcal C'_N\subset \operatorname{im}(V_N)$ of the same asymptotic rate $r$ such that every erasure of at most $\lfloor pN\rfloor$ qubits admits a recovery channel with vanishing diamond-norm error. \end{corollary}
\begin{proof} By Assumption~\ref{assump:gibbs-spectral-aep}, the proof of Theorem~\ref{thm:general-scrooge} applies uniformly over $E\in\mathcal E_{p,N}$ and gives \begin{equation} \max_{E\in\mathcal E_{p,N}} \mathbb E\,\Delta_E(V_N) \longrightarrow0 \end{equation} whenever $r<\min\{s_\infty(\beta),s(\beta)(1-2p)\}$. Since $|\mathcal E_{p,N}|\leq2^N$, Lemma~\ref{lem:uniform-scrooge-concentration} therefore implies \begin{equation} \max_{E\in\mathcal E_{p,N}} \Delta_E(V_N) \xrightarrow{\Pr}0. \end{equation} If $F\subseteq E$, contractivity of trace norm under partial trace gives $\Delta_F(V_N)\leq\Delta_E(V_N)$. Hence \begin{equation} \delta_N := \max_{|E|\leq\lfloor pN\rfloor} \Delta_E(V_N) \xrightarrow{\Pr}0. \end{equation} For each erased region $E$, the decoupling-to-recovery argument of Sec.~\ref{sec:aqec_decoupling}, together with the Fuchs--van de Graaf inequalities, gives a recovery channel $\mathcal R_{N,E}$ satisfying \begin{equation} 1- F_e\left( \frac{\mathbbm 1_{K_N}}{K_N}, \mathcal R_{N,E}\circ\mathcal N_{N,E}\circ\mathcal E_N \right) \leq \Delta_E(V_N) \leq \delta_N. \end{equation} Thus the $K_N$-dimensional code is uniformly average-case correctable against the family of all erasures of at most $\lfloor pN\rfloor$ qubits, whose cardinality is at most $2^N$. Theorem~\ref{thm: worst_to_avg} therefore yields a single subcode $\mathcal C'_N\subset\operatorname{im}(V_N)$, of the same asymptotic rate $r$, that is uniformly worst-case correctable against every such erasure. In particular, there exist recovery channels $\mathcal R'_{N,E}$ such that \begin{equation} \max_{|E|\leq\lfloor pN\rfloor} \left\| \mathcal R'_{N,E}\circ \mathcal N_{N,E}\circ \mathcal E'_N - \operatorname{id} \right\|_\diamond \xrightarrow{\Pr}0. \end{equation} 
\end{proof}

We now show that this threshold is tight. Suppose additionally that there exist disjoint regions $A_N,C_N\in\mathcal E_{p,N}$ such that, for $D_N:=[N]\setminus(A_N\cup C_N)$, the reduced Gibbs state $(g_{\beta,N})_{D_N}$ has an entropy-concentrated spectrum satisfying \begin{equation} S\left((g_{\beta,N})_{D_N}\right) = s(\beta)|D_N|+o(N). \end{equation} Let $\overline\rho_N := \frac{V_NV_N^\dagger}{K_N}$  be the maximally mixed state on the Scrooge code, and let $\Pi_{D_N}$ be the corresponding entropy-typical projector. The Scrooge first-moment identity, Markov's inequality, and Gram-matrix concentration imply \begin{equation} \operatorname{tr}\left( (\mathbbm 1_{D_N}-\Pi_{D_N}) (\overline\rho_N)_{D_N} \right) \xrightarrow{\Pr}0. \end{equation} Since $\log_2\operatorname{rank}\Pi_{D_N} \leq s(\beta)|D_N|+o(N)$, the typical-subspace entropy bound gives \begin{equation} S\left((\overline\rho_N)_{D_N}\right) \leq s(\beta)(1-2p)N+o(N) \end{equation} with probability tending to one. Consequently, \begin{equation} \min_{\substack{A\cap C=\emptyset\\ |A|=|C|=\lfloor pN\rfloor}} S\left( (\overline\rho_N)_{[N]\setminus(A\cup C)} \right) \leq s(\beta)(1-2p)N+o(N). \end{equation} Combining this bound with the simultaneous decoupling established in Corollary~\ref{corr:thermal-scrooge-uniform}, Corollary~\ref{corr:entropic_singleton_volume_law} implies \begin{equation} r\leq s(\beta)(1-2p), \end{equation} or equivalently $p\leq p_\star(r)$. Thus, thermal Scrooge codes asymptotically saturate the entropic Singleton bound in the uniform-volume-law regime.

\section{Thermalization Codes}
\label{sec:thermalization_codes}
In this section, we describe the numerical framework used to study information-theoretic optimality in thermalizing quantum dynamics. We first introduce the thermal code construction and the chaotic Hamiltonian models used in our simulations. We then describe the finite-size scaling analysis used to identify the information phase transition. Finally, we compare these results with Haar-random and Scrooge-random codes to identify the entropic mechanism underlying the transition.

\subsection{Chaotic Hamiltonians and Codeword Construction} \label{sec:numerical_therm_hamiltonians}
For $H$ a many-body Hamiltonian on N qubits (to be specified below), we let $U(t) = e^{-iHt}$ denote the time evolution unitary under $H$, which will act as our encoder. Starting from a set of orthonormal product states $\{\ket{\psi_k}\}_{k=1}^K$ we define the encoded codewords $\ket{\psi_k(t)} = U(t)\ket{\psi_k}$ which span the codespace $\mathcal{C}_t :=
    \operatorname{span}
    \left\{
        \ket{\psi_k(t)}
\right\}_{k=1}^{K}$. The encoded states are evolved to late times of the scale $t \sim N^2/J$, where $J$ is the characteristic energy scale of each model (Table~\ref{tab:hamiltonians}), well beyond microscopic relaxation timescales of the systems considered. We choose the initial states to lie in a narrow energy window, defined below, and to
carry indistinguishable values of the physically relevant extensive
conserved charges. This prevents the codeword label from being directly recorded in conserved quantities. More precisely, for any conserved charge $Q^{(a)}$ satisfying $[Q^{(a)},H]=0$ we require $P_{\mathcal C_t} Q^{(a)} P_{\mathcal C_t} \approx q_a P_{\mathcal C_t}$, where $P_{\mathcal C_t}$ is the projector onto the codespace. This is equivalent to the same condition on the initial states and can therefore be imposed when selecting the initial product states. The approximation is controlled by the width of the corresponding charge window.

We study thermalization codes generated by three nonintegrable
Hamiltonians: a mixed-field Ising model (MFIM), a Rydberg chain, and a random-field $J_1J_2$ Heisenberg chain. For the mixed-field Ising and Rydberg models, energy is the only
nontrivial extensive conserved quantity. For the
random-field $J_1J_2$ Heisenberg model, we additionally restrict to a fixed sector of the total magnetization $S^z_{\mathrm{tot}} = \sum_{i=1}^N S_i^z$. The Hamiltonians are 

\begin{align}
    H_{\mathrm{MFIM}} &= J\sum_{i=1}^{N-1} \sigma^x_i \sigma^x_{i+1} + h\sum_{i=1}^{N} \sigma^x_i + g\sum_{i=1}^{N} \sigma^y_i + h_1\sigma^x_1 - h_L\sigma^x_N, \\
    H_{\mathrm{Ryd}} &= \frac{\Omega}{2}\sum_{i=1}^{N}\sigma^x_i - \Delta\sum_{i=1}^{N} n_i + \sum_{1\le i<j\le N}\frac{C_6}{\left(|i-j|\,a\right)^6}\,n_i n_j + h_L n_1 + h_R n_N, \\
    H_{\mathrm{RF-J_1J_2}} &= \sum_{i=1}^{N-1}\!\left[S^x_i S^x_{i+1} + S^y_i S^y_{i+1} + S^z_i S^z_{i+1}\right] + \frac{1}{2}\sum_{i=1}^{N-2}\!\left[S^x_i S^x_{i+2} + S^y_i S^y_{i+2} + S^z_i S^z_{i+2}\right] + \sum_{i=1}^{N} h_i S^z_i,
\end{align}
where $\sigma_i^x, \sigma_i^y, \sigma_i^z$ denote the Pauli $X,Y,Z$ operators on site $i$, $S_i^\alpha = \sigma_i^\alpha/2$ are the corresponding spin operators, and $n_i = (1-\sigma^z_i)/2$ is the Rydberg occupation number. The model parameters are listed in Table~\ref{tab:hamiltonians}, and in each case are chosen in nonintegrable parameter regimes \cite{rodriguez-nieva_quantifying_2024}. The asymmetric boundary fields and, in the
random-field model, the spatially varying longitudinal fields remove
residual spatial symmetries that would otherwise constrain the
dynamics. For the random field $J_1J_2$ Heisenberg chain, the interior disorder fields $h_i$ are drawn from $\mathrm{Uniform}[-0.05, +0.05]$, and we fix a single disorder realization for each system-size family.

\begin{table}[h]
\centering
\setlength{\tabcolsep}{12pt}
\begin{tabular}{lll}
\hline
\textbf{MFIM} & \textbf{Rydberg} & \textbf{Random Field $J_1J_2$ Heisenberg} \\
\hline
$J = 1.0$ & $\Omega = 0.70$ & $h_1 = +0.25$ \\
$h = 0.8090$ & $\Delta = 0.50$ & $h_N = -0.137$ \\
$g = 0.9045$ & $C_6 = 1.0$ & $h_{2\le i\le N-1} \sim \mathrm{Uniform}[-0.05,+0.05]$ \\
$h_1 = 0.25$ & $a = 1.0$ & \\
$h_L = 0.25$ & $h_L = 0.137$ & \\
 & $h_R = 0.241$ & \\
\hline
\end{tabular}
\caption{Parameters used in the numerics for the three chaotic Hamiltonians.}
\label{tab:hamiltonians}
\end{table}

The code construction follows the same general procedure for all three
models. We choose a product basis adapted to the Hamiltonian and
define the diagonal energy $E(s)
    :=
\bra{s}H\ket{s}$. For a target inverse temperature $\beta$, we choose $E_\star(\beta) := \tr(Hg_\beta)$, $g_\beta
:= \frac{e^{-\beta H}}{\tr(e^{-\beta H})}$, with the trace restricted to the relevant conserved-charge sector
when necessary. We then retain product states satisfying $\left|E(s)-E_\star\right|\le \delta E_{\rm sh}$, where $\delta E_{\rm sh}$ is the chosen shell half-width. Within this shell we build a \emph{move graph} whose vertices are the shell states and whose edges connect any pair $\ket{s}, \ket{s'}$ with nonzero off-diagonal matrix element $\bra{s}H\ket{s'} \neq 0$. An independent set of this graph yields product states $\{\ket{\psi_k}\}$ of (approximately) equal energy with vanishing mutual coherence, so that $\|P_{\mathcal C_t} H P_{\mathcal C_t}  - E_\star P_{\mathcal C_t}\|_{\rm op} \le \delta E_{\rm sh}$. Thus the move-graph construction eliminates off-diagonal matrix
elements of $H$ within the codespace, while the shell restriction
controls the variation of the diagonal energies. The encoded
codewords are $\ket{\psi_k(t)} = e^{-iHt}\ket{\psi_k}$, and the code dimension is chosen according to $r = \log_2(K)/N$.

The two ingredients that differ across models are (i) the product basis used and (ii) the resulting set of \emph{forbidden moves}, i.e. the pairs of basis states connected by an edge of the move graph that must therefore be excluded from the same codespace.

\begin{itemize}
    \item \textbf{MFIM.} We work in the $\sigma^x$ product basis
$\{\ket{s_x}\}$, with $s_x\in\{+,-\}^N$. The Ising interaction, longitudinal field, and boundary fields are diagonal in this basis. The only off-diagonal term is $g\sum_i\sigma_i^y$, which flips a single $\sigma^x$ spin. The move graph therefore
connects states related by a single-site flip. At infinite temperature, the target energy is $E_\star(0)=\frac{\tr H_{\mathrm{MFIM}}}{2^N}=0$, so we use a shell centered at zero energy density.
    \item \textbf{Rydberg.} We work in the $\sigma^x$ product basis
$\{\ket{s_x}\}$. The drive
$
    \frac{\Omega}{2}\sum_i\sigma_i^x
$
is diagonal in this basis. Let
$
    m
    :=
    \#\{i:s_i^x=-\}.
$
Since
$
    \bra{s_x}n_i\ket{s_x}
    =
    \frac12
$ and $
    \bra{s_x}n_in_j\ket{s_x}
    =
    \frac14,
$
the diagonal matrix element depends only on $m$:
$
    E_m
    =
    \frac{\Omega}{2}(N-2m)
    +
    C_0,
$
where
$
    C_0
    =
    -\frac{\Delta N}{2}
    +
    \frac14\sum_{i<j}V_{ij}
    +
    \frac{h_L+h_R}{2},
$ with $ V_{ij}
    :=
    \frac{C_6}{(|i-j|a)^6}
$. Thus all $\sigma^x$ product states with fixed $m$ have the same
diagonal energy. Within a fixed-$m$ shell, the $\sigma_i^z\sigma_j^z$ component of $V_{ij}n_in_j$ exchanges a
$+$ and a $-$ spin at sites $i$ and $j$, with matrix element
$V_{ij}/4$. Since $V_{ij}\neq0$ for every pair $i<j$, the
within-shell move graph is the Johnson graph $J(N,m)$. We select
the codewords from an independent set of this graph. All remaining
off-diagonal terms change $m$ and therefore connect different
shells.
        \item \textbf{Random Field $J_1J_2$ Heisenberg.} We work in the $S^z$ product basis and restrict to a fixed
$S^z_{\mathrm{tot}}$ sector, choosing
$S^z_{\mathrm{tot}}=0$ for even $N$. The off-diagonal exchange
operator
$S_i^xS_j^x+S_i^yS_j^y
    =
    \frac12
    \left(
        S_i^+S_j^-
        +
        S_i^-S_j^+
    \right)$
swaps antiparallel spins at sites $i$ and $j$. The move graph
therefore connects states related by nearest-neighbor swaps
$(i,i+1)$ or next-nearest-neighbor swaps $(i,i+2)$.
\end{itemize}

For the energy-separated codewords of Fig.~2 in the main text, we use the same construction with two shells rather than one. We select orthogonal product states $\ket{\psi_1}$ and $\ket{\psi_2}$ from shells centered at $\pm \Delta E_N/2$, so that $\bra{\psi_1}H\ket{\psi_1} - \bra{\psi_2}H\ket{\psi_2} = \Delta E_N$.
\subsection{Finite-Size Scaling Analysis}

As $p$ increases, the code undergoes an information  transition: an \textit{error-correcting} phase for $p < p_\star$, where $I(E:R)$ vanishes and the logical information is recoverable from the complement, and a \textit{revealing} phase for $p > p_\star$, where $I(E:R)$ grows and the environment gains information about the logical state. 

Here we describe how we extract the threshold noise rate $p_\star(r)$ at fixed asymptotic code rate
$r$. For each system size $N$, we choose an integer number of logical qubits $k_N$ and code dimension $K_N=2^{k_N}$, and finite rate $r_N=\frac{k_N}{N}$ such that $r_N\to r$ as $N\to\infty$. For notational simplicity, we label each finite-size family by its target rate $r$. We evolve the system $S$ under the Hamiltonian $H$ to obtain the late-time codewords $\ket{\psi_k(t)}$, which span the encoded subspace $\mathcal{C}_t=\operatorname{span}\{{\ket{\psi_k(t)}}_{k=1}^{K_N}\}$. The encoded logical state is purified by a $K_N$-dimensional reference system $R$. We then apply a flagged erasure channel to $S$, in which a fraction $p$ of the physical qubits is transferred to an environment $E$, with the erasure locations recorded by classical flags. Using the decoupling
protocol, as described in Section~\ref{sec:aqec_decoupling}, we evaluate the approximate error-correction diagnostic
$I_N(p;r) := I(E:R)$, the quantum mutual information between the channel environment $E$ and the reference
$R$, where the subscript $N$ denotes the finite-size dependence.

We repeat this procedure over various rates $r$, noise strengths $p$, and system sizes $N$. Near the transition, we assume the leading finite-size scaling form Eq.~\eqref{eq:fss_ansatz}, with $\nu$ controlling the width and $\mu$ the scaling of the mutual information within it:
\begin{equation}
I_N(p;r)=N^{-\mu}\,f\!\left((p-p_\star(r))\,N^{1/\nu}\right),
\label{eq:fss_ansatz}
\end{equation}
so that the rescaled curves $N^{\mu}I_N(p;r)$ collapse onto a
single function $f((p-p_\star(r))N^{1/\nu})$ in the scaling regime. 

We fit in two stages because a joint fit of $p_\star(r)$ and $\nu$ is poorly conditioned -- errors in one can be absorbed by the other. Thus, first, we simultaneously determine $p_\star(r)$ and the scaling dimension $\mu$ from the crossing condition. At $p = p_\star (r)$, Eq.~\eqref{eq:fss_ansatz} gives \begin{equation} N^\mu I_N(p_\star(r);r)=f(0),\end{equation} which is independent of $N$: the rescaled curves $N^\mu I_N(p_\star(r)$ for different system sizes intersect at the transition point. We therefore minimize \begin{equation} Q_1(p,\mu) = \sum_{i<j} \left[ \frac{ \left| N_i^\mu I_{N_i}(p;r) - N_j^\mu I_{N_j}(p;r) \right| }{ \max\!\left( \left|N_i^\mu I_{N_i}(p;r)\right|, \left|N_j^\mu I_{N_j}(p;r)\right|, \epsilon_{\rm reg} \right)} \right]^2 \label{eq:stage1_cost} \end{equation} where $\epsilon_{\rm reg}>0$ is a small regularizer preventing division by zero. The fitted threshold and scaling dimension are \begin{equation} \bigl(p_\star(r),\mu\bigr) = \arg\min_{p,\mu} Q_1(p,\mu). \end{equation} With $p_\star(r)$ and $\mu$ fixed, $\nu$ is obtained by minimizing 
\begin{equation} Q_2(\nu) = \sum_{i<j} \left[ \frac{ \sum_{x\in\Omega_{ij}(\nu)} \left| N_i^\mu I_{N_i}\!\left(p_\star(r)+xN_i^{-1/\nu};r\right) - N_j^\mu I_{N_j}\!\left(p_\star(r)+xN_j^{-1/\nu};r\right) \right| }{ \max\!\left( \sum_{x\in\Omega_{ij}(\nu)} \left|N_i^\mu I_{N_i}\!\left(p_\star(r)+xN_i^{-1/\nu};r\right)\right|, \sum_{x\in\Omega_{ij}(\nu)} \left|N_j^\mu I_{N_j}\!\left(p_\star(r)+xN_j^{-1/\nu};r\right)\right|, \epsilon_{\rm reg} \right)} \right]^2 . \label{eq:stage2_cost} \end{equation} The collapse exponent is $\nu := \arg \min_\nu Q_2(\nu)$. Here $\Omega_{ij}(\nu)$ is a discrete grid in the common overlap of the rescaled coordinates for the pair $(N_i,N_j)$ after fit-window cutoffs are applied. The values of $I_N$ entering Eqs.~\eqref{eq:stage1_cost} and~\eqref{eq:stage2_cost} are obtained by interpolation of the finite-size data. We require a minimum number of interpolation points in each overlap window, since $Q_2$ would otherwise favor values of $\nu$ that simply reduce to the number of points being compared.

Before performing the fits, we restrict the data to the scaling regime. Two cutoffs are applied. First, for each system size $N$ let
$p_{\min}(N;r)$ be the largest sampled value of $p$ satisfying $I_N(p;r)\le0.05$. We discard all data with $p<p_{\min}(N;r)$. This removes the flat low-noise tail where $I(E:R)\approx 0$, which would otherwise artificially improve the collapse by matching zero to zero. Second, we impose an upper cutoff $p_{\max}=0.55$ which restricts the fit to the transition region of interest. To characterize the sensitivity of the fitted threshold, we fix $\mu$ at the fitted value and scan $Q_1(p,\mu)$ on a fine grid of $p$ within the allowed interval. We report the range satisfying  $Q_1(p,\mu)\leq 2 Q_1^{\min}$, and take half of its width as the quoted uncertainty in $p_\star(r)$. This error bar reflects the width of the residual landscape near the optimum.

\subsection{Entropic Origin of the Information Transition}

The preceding sections establish that $g_\beta$-Scrooge random codes achieve the entropic Singleton threshold, but the arguments run through decoupling inequalities rather than entropies themselves. A direct calculation of the subsystem entropies gives a more transparent account of the transition, and reproduces the observed finite-size behavior. We therefore analyze the subsystem entropies constituting the mutual information $I(E:R) = S(E) + S(R) - S(ER)$ where $E$ is the erased subsystem, $B$ is the complement, and $R$ is the logical reference. We begin at infinite-temperature with Haar-random codes. 

Let  $\mathcal H=\mathcal H_A\otimes \mathcal H_{\bar A}$ be a bipartition into a generic subsystem $A$ and its complement, with $d_A=\dim \mathcal H_A$, $d_{\bar A}=\dim \mathcal H_{\bar A}$ and $d=d_A d_{\bar A}=2^N$. Suppose that the codewords $\{|\bar\psi_k\rangle\}_{k=1}^K$ are the first $K$ columns of a Haar-random unitary on $\mathcal H$. We consider the mixed reduced state on $A$, $\rho_A^{\rm mix}
=
\tr_{\bar A}\left[
\frac1K\sum_{k=1}^K
|\bar\psi_k\rangle\langle \bar\psi_k|
\right]$. We claim, in the regime $K/d \rightarrow 0$, the mixed reduced code state is well approximated at leading order by the reduced density matrix of a single Haar-random state on the enlarged Hilbert space $\mathcal H_A\otimes \mathcal H_{\bar A}\otimes \mathbb C^K $. To see this, let $Y\in\mathbb C^{d\times K}$ have independent
standard complex Gaussian entries. The polar factor $V:= Y(Y^\dagger Y)^{-1/2}$ is a Haar-random isometry from   $\mathbb C^K$ into $\mathcal H$. The mixed code state before tracing out $\bar A$ is therefore
\begin{equation}
    \omega_{\rm code}
=
\frac1K VV^\dagger
=
\frac1K Y(Y^\dagger Y)^{-1}Y^\dagger .
\end{equation}
On the other hand, interpreting the column index of $Y$ as an additional reference index $r=1,\dots,K$, then the normalized Gaussian vector\begin{equation}
    |\Gamma_K\rangle
=
\frac{1}{\sqrt{\tr(Y^\dagger Y)}}
\sum_{a=1}^{d_A}
\sum_{\bar a=1}^{d_{\bar A}}
\sum_{r=1}^{K}
Y_{(a,\bar a),r}
|a\rangle_A |\bar a\rangle_{\bar A} |r\rangle_R
\end{equation}
is a Haar-random state on $\mathcal H_A\otimes \mathcal H_{\bar A}\otimes \mathbb C^K$. 
Its reduced state on $A$ is \begin{equation}
    \sigma_A
=
\tr_{\bar A R}\dyad{\Gamma_K}
=
\tr_{\bar A}\left(
\frac{YY^\dagger}{\tr(Y^\dagger Y)}
\right).
\end{equation}
For a $d \times K$ Gaussian matrix with $K/d \to 0$, the Wishart matrix $Y^\dagger Y$ concentrates around $d I_K$, $\frac{1}{d}Y^\dagger Y \approx I_K$. Hence, \begin{equation}
    Y(Y^\dagger Y)^{-1}Y^\dagger
\approx
\frac{1}{d}YY^\dagger. 
\end{equation}
After normalization and tracing out $\bar A$, this gives
\begin{equation}
    \rho_A^{\rm mix}
\approx
\sigma_A
=
\tr_{\bar A R}(\dyad{\Gamma_K})
\end{equation}

Thus, in the regime $K\ll 2^N$, the entropy of the mixed reduced state is approximated by the Page entropy \cite{page_average_1993} of a Haar-random state on $\mathcal H_A\otimes (\mathcal H_{\bar A}\otimes \mathbb C^K)$,
\begin{equation}
    S(\rho_A^{\rm mix}) \approx \begin{cases} \log_2(K d_{\bar{A}}) - \dfrac{K d_{\bar{A}}}{2 d_A \ln 2}, & K d_{\bar{A}} \leq d_A, \\[8pt] \log_2 d_A - \dfrac{d_A}{2K d_{\bar{A}} \ln 2}, & K d_{\bar{A}} \geq d_A \end{cases}.
\end{equation}

The shifted Page curve is controlled by the smaller of the subsystem
dimension $d_A$ and the effective complementary dimension $K d_{\bar A}$,
with a transition when $d_A = K d_{\bar A}.$ For $d_A \leq K d_{\bar A}$, the entropy is close to its maximal value
$\log_2 d_A$, while for $d_A \geq K d_{\bar A}$ it is close to
$\log_2(K d_{\bar A})$. The subleading terms are the usual Page
corrections with the replacement $d_{\bar A}\rightarrow Kd_{\bar A}$. They are exponentially small in $N$ away from the crossover $Kd_{\bar A} = d_A$, and $O(1)$ there.

Since the global state on $EBR$ is pure, $S(ER)=S(B)$, while
$S(R)=\log_2K=rN$. Applying the shifted Page formula with $A = B$, for which
$d_B=2^{(1-p)N}$ and the effective complementary dimension is
$K d_E=K2^{pN}$, its Page transition occurs when $ 2^{(1-p)N}=K2^{pN}$. This gives $p_\star=\frac{1-r}{2}$. Near this transition, $S(E)\simeq pN$, whereas
\begin{equation}
    S(B)\approx
    \begin{cases}
        (p+r)N, & p<p_\star,\\
        (1-p)N, & p>p_\star
    \end{cases}.
\end{equation}
Consequently,
\begin{equation}
    I(E:R)
    =
    S(E)+S(R)-S(ER)
    =
    S(E)+\log_2K-S(B)
\end{equation}
vanishes at leading order below $p_\star$ and grows as $I(E:R)\simeq 2N(p-p_\star)$ immediately above it. The Page correction is $O(1)$ in absolute entropy
and therefore rounds the transition over a window
$p-p_\star=O(N^{-1})$, corresponding to $\nu=1$. Since the raw mutual
information remains $O(1)$ within this critical window, its scaling
dimension is $\mu=0$.

Just as a Haar-random encoding provided a useful analytical tool to study the information phase transition for our infinite-temperature thermalization codes, we use Scrooge-random codes in the same way for finite temperature thermalization codes. Let $V=\Phi G^{-1/2}$ be the encoding isometry for the $\rho$-Scrooge code as in Section~\ref{sec:scrooge_codes_def}. The mixed code state is
\begin{equation}
    \omega_{\rm code}
    =
    \frac1K VV^\dagger
    =
    \frac1K\Phi G^{-1}\Phi^\dagger.
\end{equation}
Lemma~\ref{lem:gram} shows the Gram matrix concentrates near $I_K$ with high probability in the low-purity regime. Hence, on the event $\|G-I_K\|_{\rm op}\le \eta<1$, we have $G^{-1}=I_K+O(\eta)$, so the symmetric orthogonalization gives only a subleading correction. The mixed state of the orthogonalized codewords is
\begin{equation}
    \omega_{\rm code}
\approx
\frac{1}{K}\Phi\Phi^\dagger
=
\frac{1}{K}\sum_{i=1}^K
\dyad{\phi_i}.
\end{equation}
Let $\Pr_0$ denote the law of $K$ independent standard complex
Gaussian vectors $\ket{g_1},\ldots,\ket{g_K}$, and define
\begin{equation}
    q_i=\bra{g_i}\rho\ket{g_i},
    \qquad
    X=\sqrt{\rho}\,[g_1\;\cdots\;g_K],
    \qquad
    D_X=\operatorname{diag}(q_1,\ldots,q_K).
\end{equation}
The joint Gaussian measure corresponding to $K$ independent samples
from ${\rm Scr}(\rho)$ is
\begin{equation}
    \frac{\mathrm d\Pr_{\rm prod}}{\mathrm d\Pr_0}
    =
    \prod_{i=1}^K q_i,
\end{equation}
under which
\begin{equation}
    \ket{\phi_i}
    =
    \frac{\sqrt{\rho}\ket{g_i}}{\sqrt{q_i}},
    \qquad
    \omega_{\rm code}
    \approx
    \frac{1}{K}XD_X^{-1}X^\dagger.
\end{equation} By contrast, define
\begin{equation}
    \ket{\Lambda_K}
    =
    \frac{
        \sum_{i=1}^K\sqrt{\rho}\ket{g_i}\otimes\ket{i}_R
    }{
        \sqrt{\sum_{i=1}^K q_i}
    }.
\end{equation}
This state is distributed according to
${\rm Scr}(\rho\otimes I_K/K)$ under the globally tilted measure
\begin{equation}
    \frac{\mathrm d\Pr_{\rm glob}}{\mathrm d\Pr_0}
    =
    \frac{1}{K}\sum_{i=1}^K q_i,
\end{equation}
and its reduced state is
\begin{equation}
    \omega_{\rm enl}
    =
    \tr_R\dyad{\Lambda_K}
    =
    \frac{XX^\dagger}{\tr(X^\dagger X)}.
\end{equation}
The measures $\Pr_{\rm prod}$ and $\Pr_{\rm glob}$ are not
identical. However, writing $\tau=\tr(\rho^2)$,
\begin{equation}
    \left\|\Pr_{\rm prod}-\Pr_{\rm glob}\right\|_{\rm TV}
    \leq
    \frac{1}{2}\sqrt{(1+\tau)^K-1}
    +
    \frac{1}{2}\sqrt{\frac{\tau}{K}},
\end{equation}
so they become asymptotically equivalent whenever
$K\tr(\rho^2)\to0$. If, in addition, the Gram-matrix concentration
conditions of Lemma~\ref{lem:gram} hold, then $X^\dagger X\approx I_K$ under the relevant measures, so that $D_X\approx I_K$ and
$\tr(X^\dagger X)\approx K$. Under these combined conditions, $\omega_{\rm code} \approx \omega_{\rm enl}$. Thus, in the regime where both the measure-comparison and
Gram-concentration errors vanish, mixing over $K$ approximately
orthogonal Scrooge codewords is asymptotically equivalent to appending
a $K$-dimensional reference register to the complement of any subsystem
$A$. In particular, $\rho_A^{\rm mix}
    \approx
    \tr_{\bar A R}(
    \dyad{\Lambda_K})$, where $\ket{\Lambda_K}$ is Scrooge-random with parent density matrix
$\rho\otimes I_K/K$. At infinite-temperature,
$\rho=I_d/d$, the Scrooge ensemble reduces to the Haar measure and
$\rho\otimes I_K/K=I_{dK}/(dK)$, recovering the previous result.
 
However, away from $I_d/d$, the subsystem entropy of a
Scrooge-random state is not determined only by the dimensions $d_A, d_{\bar A}$. It also depends on the structure of $\rho$, including its eigenbasis relative to the bipartition $A |\bar A$. Thus a simple Page-like closed-form expression for general $\rho$ does not exist. Let \begin{equation}
    \chi_n^{\rm Scr}(\rho)
\coloneqq
\mathbb E_{\phi\sim{\rm Scr}(\rho)}
\left[\dyad{\phi}\right]^{\otimes n}
\end{equation} 
denote the $n$-th moment operator of the Scrooge ensemble. Formulae for $\chi_n^{\rm Scr}(\rho)$ were derived in \cite{mark_maximum_2024}.  The subsystem moments are then obtained by the swap trick, 
\begin{equation}
    M_n^{\rm Scr}(\rho;A)
=
\tr\left(
\left(F_A^{(n)}\otimes I_{\bar A}^{\otimes n}\right)
\chi_n^{\rm Scr}(\rho)
\right),
\end{equation}
where $F_A^{(n)}$ cyclically permutes the $n$ replicas of subsystem $A$. The averaged von Neumann entropy is formally obtained by the replica limit
\begin{equation}
    \overline S_A^{\rm Scr}(\rho)
=
-\frac{1}{\ln 2}\left.\partial_n M_n^{\rm Scr}(\rho;A)\right|_{n=1}.
\end{equation}

Therefore, for the mixed Scrooge code, the analogue of
the Haar Page benchmark is $\overline S_A^{\rm Scr}\!\left(\rho\otimes \frac{I_K}{K}\right)$ where the complement of $A$ is enlarged from $\bar A$ to $\bar A R$, with $R\simeq \mathbb C^K$.

When $\rho$ is specialized to a Gibbs state,
$\rho=g_\beta\coloneqq e^{-\beta H}/\tr(e^{-\beta H})$, the Scrooge construction
is closely related to, but distinct from, canonical thermal pure quantum
(cTPQ) states \cite{sugiura_canonical_2013}. Both constructions use the
thermal filtering map
\begin{equation}
    \ket{\psi}
    \longmapsto
    \frac{\sqrt{g_\beta}\ket{\psi}}
    {\sqrt{\bra{\psi}g_\beta\ket{\psi}}},
\end{equation}
but cTPQ states take $\ket{\psi}$ to be Haar-random, whereas the Scrooge
ensemble additionally reweights the Haar measure by
$d\bra{\psi}g_\beta\ket{\psi}$. Thus the two ensembles are not identical
at finite system size. However, this reweighting becomes asymptotically
negligible when $\mathrm{tr}(g_\beta^2)\to0$, so finite-temperature Page curves
derived for cTPQ states provide an asymptotic benchmark for the
$g_\beta$-Scrooge ensemble in the high-temperature regimes we consider here. Previous works used cTPQ states to derive
analytic R\'enyi Page-curve formulas and to obtain the von Neumann Page
curve via replica methods in the high-temperature regime, supported by numerical comparisons
\cite{nakagawa_universality_2018,fujita_page_2018}. Motivated by the finite-temperature Page curves of cTPQ
states~\cite{nakagawa_universality_2018,fujita_page_2018}, we model the
enlarged $g_\beta$-Scrooge state by a uniform thermal entropy density
$s(\beta)$, measured in bits per site. At leading order, the entropy of $B$ is
\begin{equation}
    S(B)
    \simeq
    \min
    \left\{
        s(\beta)(1-p)N,\,
        s(\beta)pN+rN
    \right\}.
    \label{eq:finite-temperature-Page-branches}
\end{equation}
The two terms exchange dominance when $s(\beta)(1-p)N
    =
    s(\beta)pN+rN$, giving
\begin{equation}
    p_\star
    =
    \frac{s(\beta)-r}{2s(\beta)}.
\end{equation}
This threshold is meaningful for $0\le r<s(\beta)$. Assuming that the finite-temperature Page correction remains $O(1)$
in absolute entropy near the branch crossing, the transition is rounded
over a window $p-p_\star
    =
    O\left(
        1/(s(\beta)N)
    \right)$. For fixed $\beta$, $s(\beta)$ is independent of $N$, so this
corresponds to the horizontal scaling exponent $\nu=1$. The mutual
information remains $O(1)$ within this window, giving
$\mu=0$.

\section{Error Correction Subject to Conservation Constraints}
\label{sec:conservation_constraints}
A natural obstruction to thermalization acting as an error correcting code arises when the codewords differ in conserved quantities, as illuminated by Observation 1. Codewords with different energies generally equilibrate to thermal states with different effective temperatures, allowing the environment $E$ with access to an erased subsystem to retain information about the logical state through measurements of local energy density. 

In this section, we characterize the limits of error correction in thermalizing dynamics by determining when such information leakage occurs. We first compare the full quantum mutual information $I(E:R)$ with the information accessible through measurements of the conserved quantity on $E$. For energy-separated codewords, we show that this \textit{measurement-restricted mutual information} reduces to a Holevo information~\cite{holevo_bounds_1973}, while its fully classical component is the Jensen-Shannon divergence between the codeword-dependent subsystem energy distributions. We then use a high-temperature Gaussian model to determine how this classical leakage scales with the codeword energy separation in the thermodynamic limit. Finally we extend this analysis to several conserved quantities and examine how the leakage and thus error correcting regime depends on the specific noise channel as well.

\subsection{Single Conserved Quantity}\label{sec:single-charge}

We first isolate the information leakage associated with a single conserved quantity. Consider a two-dimensional thermalization code  $\mathcal{C}_t = \text{span}\{\ket{\psi_1(t)}, \ket{\psi_2(t)}\}$ where the codewords have different conserved energies. For concreteness, we focus on energy, though the discussion extends straightforwardly to other conserved quantities (e.g., charge). After shifting the zero of energy, we write $\langle H\rangle_{\psi_{1,2}(t)} = \pm \Delta E_N/2$, with respect to the encoding Hamiltonian $H$ at system size $N$. 

As in the decoupling setup of Section~\ref{sec:aqec_decoupling}, we decompose the physical system as $S=EB$, where $E$ is the subsystem accessible to the environment and $B$ is its complement. For the two-dimensional code, the encoded purification is
\begin{equation} \label{eq:decoup_enc_pur_init}
    \ket{\Psi(t)}_{RS}
=
\frac{1}{\sqrt{2}}
\sum_{i=1}^{2}
\ket{i}_R\ket{\psi_i(t)}_S,
\qquad
\sigma_{RS}(t)
:=
\dyad{\Psi(t)}_{RS}.
\end{equation}
The corresponding state on the reference and environment, $RE$, is 
\begin{equation}
\sigma_{RE}(t)
=
\tr_B \sigma_{RS}(t)
=
\frac{1}{2}
\sum_{i,j=1}^{2}\dyad{i}{j}\otimes
\sigma_E^{(ij)}(t),
\end{equation}
where $\sigma_E^{(ij)}(t)
:=
\tr_B
\left(
\dyad{\psi_i(t)}{\psi_j(t)}
\right)$. Let $H_E$ denote the Hamiltonian restricted to $E$, and assume for simplicity that its spectrum is nondegenerate, $H_E = \sum_e e \Pi_e$, $\Pi_e :=\dyad{e}$. Resolving the energy of $E$ is equivalent to applying the dephasing channel $\mathcal D_E(\cdot) = \sum_e \Pi_e(\cdot)\Pi_e$. We define the \textit{measurement-restricted mutual information}, $I(E:R|\{\Pi_e\}):= I(E:R)_{
(\operatorname{id}_R \otimes \mathcal D_E)(\sigma_{RE})
}$. Because $\mathcal D_E$ is a local channel on $E$, data processing gives, $I(E:R) \ge I(E:R|\{\Pi_e\})$. 
Applying the local energy dephasing channel to $E$ gives
\begin{align}
\left(
\operatorname{id}_R\otimes\mathcal D_E
\right)
\bigl(\sigma_{RE}(t)\bigr)
&=
\frac{1}{2}
\sum_{i,j=1}^{2}
\dyad{i}{j}_R
\otimes
\sum_e
\Pi_e\sigma_E^{(ij)}(t)\Pi_e
\\
&=
\sum_e
p(e)
\sigma_R^{(e)}(t)\otimes\dyad{e}_E,
\end{align}
where,
\begin{align}
p(e)
:=
\frac{1}{2}
\sum_{i=1}^{2}
\bra{e}\sigma_E^{(ii)}(t)\ket{e},\qquad
\sigma_R^{(e)}(t)
:=
\frac{1}{2p(e)}
\sum_{i,j=1}^{2}
\bra{e}\sigma_E^{(ij)}(t)\ket{e}
\dyad{i}{j}_R.
\end{align}
The dephased state is therefore classical on $E$ and quantum on $R$. Consequently, $I(E:R|\{\Pi_e\})$ is the Holevo information \cite{holevo_bounds_1973} of the ensemble
$\{p(e),\sigma_R^{(e)}\}_e$. This quantity captures all correlations with the reference that remain after coherences between distinct subsystem-energy sectors have been removed.

To isolate the purely classical information carried by the codeword-dependent energy distribution, we additionally dephase $R$ in the codeword-label basis $\mathcal D_R(\cdot)
:=
\sum_{i=1}^{2}
\dyad{i}(\cdot)\dyad{i}$. The resulting classical state is
\begin{equation}
    \left(
\mathcal D_R\otimes\mathcal D_E
\right)
\bigl(\sigma_{RE}(t)\bigr)
=
\sum_{i,e}
p(i,e)
\dyad{i}_R\otimes\dyad{e}_E, \qquad p(i,e)
=
\frac{1}{2}
\bra{e}\sigma_E^{(ii)}(t)\ket{e}.
\end{equation}
Since the codewords occur with equal probability $p(i) = 1/2$, their conditional subsystem energy distributions are $p^{(i)}(e)
:=
p(e\mid i)
=
\bra{e}\sigma_E^{(ii)}(t)\ket{e}$. Each $p^{(i)}$ is a probability distribution over the subsystem-energy outcomes $e$, giving the distribution of local energies measured on $E$ when the codeword is $\ket{\psi_i(t)}$. The corresponding classical mutual information is the equal-prior Jensen--Shannon divergence (JSD), 
\begin{align} \label{eq:MI-is-JS}
    I(E:R)_{
(\mathcal D_R\otimes\mathcal D_E)(\sigma_{RE})
}
=
H\left(
\frac{p^{(1)}+p^{(2)}}{2}
\right)
-
\frac{1}{2}H(p^{(1)})
-
\frac{1}{2}H(p^{(2)}) =  \mathrm{JSD}(p^{(1)},p^{(2)}).
\end{align}
Data processing gives us $I(E:R) \ge I(E:R|\{\Pi_e\}) \ge \mathrm{JSD}(p^{(1)},p^{(2)})$. The full mutual information measures all the information about the logical qubit available to $E$, the measurement-restricted mutual information retains only the information accessible after resolving the subsystem energy, and the JSD retains only the classical correlation between the energy outcome and the codeword label. The numerical agreement of these three quantities (when $p< 1/2$) as shown in Fig.~2(b) and  Fig.~2(c) in the main text, therefore indicates that the information available to $E$ is entirely classical, captured by the classical distinguishability of the subsystem-energy distributions.

We now determine how the classical JSD in Eq.~\eqref{eq:MI-is-JS} depends on the
codeword energy separation, writing $\Delta E_N = c N^\kappa$ with $c>0$
independent of $N$.

Let $e$ be the outcome of measuring $H_E$, the part of $H$ supported entirely
within $E$. At high temperature, with exponentially decaying connected
energy-density correlations, the subsystem energy is a sum of many weakly
correlated local terms, so its distribution is asymptotically Gaussian with
variance $\Theta(N)$. Taking $E$ to be $pN$ qubits in a union of $O(1)$
intervals, boundary terms contribute $O(1)$ and the two codewords inherit a
fraction $p$ of the total separation. We therefore model the
codeword-dependent distributions as
\begin{equation}
    p^{(\pm)}(e)
    \approx
    \mathcal N\!\left(\pm\frac{\delta\mu_N}{2},\, \alpha_p N\right),
    \qquad
    \delta\mu_N = p\,\Delta E_N + O(1),
    \qquad
    \alpha_p > 0,
    \label{eq:gaussian-model}
\end{equation}
where we have centered $e$ about the midpoint of the two means, using
invariance of the JSD under a common translation. The common leading variance
is natural for $\Delta E_N = o(N)$, where the codeword energy densities
coincide in the thermodynamic limit.

The JSD between two equal-variance Gaussians depends only on the
signal-to-noise ratio
\begin{equation}
    \lambda_N
    :=
    \frac{\delta\mu_N}{2\sqrt{\alpha_p N}}
    =
    \frac{pc}{2\sqrt{\alpha_p}}\,N^{\kappa-\frac12}\left(1+o(1)\right).
    \label{eq:snr}
\end{equation}
Within the Gaussian approximation,
\begin{equation} \label{eq:Gaussian-JSD}
    \operatorname{JSD}(p^{(+)},p^{(-)})
\approx
\mathcal J(\lambda_N), \qquad \mathcal J(\lambda)
:=
1-
\frac{1}{\sqrt{2\pi}}
\int_{-\infty}^{\infty}
dz
e^{-z^2/2}
\log_2
\left(
1+e^{-2\lambda z-2\lambda^2}
\right).
\end{equation}

For $\kappa < 1/2$, $\lambda_N \rightarrow 0$. Expanding around $\lambda = 0$, $\mathcal J(\lambda)
=
\frac{\lambda^2}{2\ln2}
+
O(\lambda^4)$ and therefore, 
\begin{equation}
    \operatorname{JSD}(p^{(+)},p^{(-)})
=
\frac{p^2c^2}
{8\alpha_p\ln2}
N^{2\kappa-1}(1+o(1)).
\end{equation} The classical information available to $E$ therefore vanishes in the thermodynamic limit whenever $\Delta E_N=o(\sqrt N)$. For $\kappa = 1/2$, the signal-to-noise ratio approaches the finite constant $\lambda_\star
=
\frac{pc}{2\sqrt{\alpha_p}}$ so that $\operatorname{JSD}(p^{(+)},p^{(-)})
\longrightarrow
\mathcal J(\lambda_\star)
\in(0,1)$. This is the crossover regime in which the $O(\sqrt N)$ separation between the codeword energies is comparable to the $O(\sqrt N)$ thermal fluctuations of the subsystem.
For $\kappa>1/2$, the separation between the two subsystem-energy means grows faster than their $O(\sqrt{N})$ widths. Within the Gaussian model of Eq.~\eqref{eq:gaussian-model}, the overlap between the two distributions therefore vanishes, so the subsystem-energy outcome asymptotically identifies the codeword label and $\operatorname{JSD}(p^{(+)},p^{(-)})
\longrightarrow 1$.

Therefore, within the high-temperature Gaussian model, an energy measurement on a fixed-fraction subsystem reveals vanishing information about the codeword label when
$\Delta E_N=o(\sqrt N)$, finite information when
$\Delta E_N=\Theta(\sqrt N)$, and asymptotically one bit when
$\Delta E_N\gg\sqrt N$. More general subsystem geometries modify the nonuniversal
signal-to-noise prefactor but preserve this crossover whenever the
subsystem mean separation remains $\Theta(\Delta E_N)$ and its variance remains $\Theta(N)$. Combined with the numerical agreement $I(E:R) \simeq I(E:R|\{\Pi_e\}) \simeq \mathrm{JSD}(p_1,p_2)$ below the information transition, this supports the conclusion that the observed information gained by the environment is governed by classical distinguishability of the subsystem-energy distributions.

\subsection{Multiple Conserved Quantities}

The single-charge analysis extends to several simultaneously measurable
conserved quantities, for example a $U(1)$-symmetric spin chain with conserved
energy and total magnetization $S^z$. We consider a two-dimensional codespace
and a fixed number, $m=O(1)$, of additive conserved quantities
$\mathbf Q := (Q^{(1)},\ldots,Q^{(m)})$ with $Q^{(a)}=\sum_{j=1}^N q_j^{(a)}$, satisfying
$[Q^{(a)},H]=0$ and $[Q^{(a)},Q^{(b)}]=0$. Assume the restrictions $Q_E^{(a)}$ commute and,
for simplicity, have a nondegenerate joint spectrum
$Q_E^{(a)}=\sum_{\mathbf q}q^{(a)}\Pi_{\mathbf q}$, with $\mathbf q=(q^{(1)},\ldots,q^{(m)})$
the joint subsystem-charge outcome.

Replacing the energy dephasing channel $\mathcal D_E$ by its joint-charge
analogue $\mathcal D_E^{\mathbf Q}(\cdot):=\sum_{\mathbf q}
\dyad{\mathbf q}(\cdot)\dyad{\mathbf q}$, the derivation of
Sec.~\ref{sec:single-charge} goes through verbatim: the dephased state is
classical on $E$ and quantum on $R$, $I(E:R\mid\{\Pi_{\mathbf q}\})$ is the
corresponding Holevo information, and dephasing $R$ in the codeword-label
basis gives
\begin{equation}
    I(E:R)
    \ge
    I(E:R\mid\{\Pi_{\mathbf q}\})
    \ge
    \operatorname{JSD}(p^{(1)},p^{(2)}),
    \qquad
    p^{(i)}(\mathbf q)
    :=
    \bra{\mathbf q}\sigma_E^{(ii)}\ket{\mathbf q}.
\end{equation}
A degenerate joint spectrum changes nothing: replacing
$\mathcal D_E^{\mathbf Q}$ by the measurement channel
$\mathcal M_E^{\mathbf Q}(\cdot):=\sum_{\mathbf q}
\tr(\Pi_{\mathbf q}(\cdot))\dyad{\mathbf q}$ leaves the distributions
$p^{(i)}(\mathbf q)$, and hence the analysis below, unchanged.

The Gaussian model generalizes in the obvious way. With
$\Delta\mathbf Q_N = N^\kappa\mathbf d + o(N^\kappa)$, $\mathbf d\neq\mathbf 0$,
the scalar Gaussians of Eq.~\eqref{eq:gaussian-model} become multivariate,
\begin{equation}
    p^{(\pm)}(\mathbf q)
    \approx
    \mathcal N\!\left(\pm\frac{\delta\boldsymbol\mu_N}{2},\, N\Sigma_p\right),
    \qquad
    \delta\boldsymbol\mu_N = p\,\Delta\mathbf Q_N + o(N^\kappa),
    \qquad
    \Sigma_p\succ0,
\end{equation}
with a common leading covariance $\Sigma_p$, and the signal-to-noise ratio is
replaced by the Mahalanobis distance between the mean vectors \cite{de_maesschalck_mahalanobis_2000},
\begin{equation}
    \lambda_N^2
    :=
    \frac{1}{4N}\,
    \delta\boldsymbol\mu_N^{\top}\Sigma_p^{-1}\delta\boldsymbol\mu_N
    =
    \frac{p^2}{4}\,
    N^{2\kappa-1}\,
    \mathbf d^{\top}\Sigma_p^{-1}\mathbf d
    \left(1+o(1)\right).
\end{equation}

The multivariate problem nonetheless reduces to a scalar one. The log-likelihood ratio,
$\log\!\left[p^{(+)}(\widetilde{\mathbf q})/p^{(-)}(\widetilde{\mathbf q})\right]
= \tfrac1N\delta\boldsymbol\mu_N^{\top}\Sigma_p^{-1}\widetilde{\mathbf q}$,
depends on the centered outcome $\widetilde{\mathbf q}$ only through
\begin{equation}
    z
    :=
    \frac{\delta\boldsymbol\mu_N^{\top}\Sigma_p^{-1}\widetilde{\mathbf q}}
    {\sqrt{N\,\delta\boldsymbol\mu_N^{\top}\Sigma_p^{-1}\delta\boldsymbol\mu_N}}
    \sim
    \mathcal N(\pm\lambda_N, 1),
\end{equation}
so the $m$-dimensional problem reduces to the scalar one with the same
$\mathcal J(\lambda_N)$ of Eq.~\eqref{eq:Gaussian-JSD}. The three regimes of
Sec.~\ref{sec:single-charge} therefore carry over unchanged, with the
crossover still at $\kappa=1/2$: a fixed number of commuting charges replaces
the scalar signal-to-noise ratio by the Mahalanobis distance, but does not
move the fluctuation scale. More general geometries alter the coefficient
relating $\delta\boldsymbol\mu_N$ to $\Delta\mathbf Q_N$, and the crossover
persists whenever $\lambda_N^2 = \Theta(N^{2\kappa-1})$.

When the conserved quantities do not commute, as for the three components of total spin in an $SU(2)$-symmetric system, they cannot be assigned simultaneous values, so there is no joint eigenbasis and no unique associated JSD. One must specify what the environment measures: a commuting subset reduces to the analysis above, while a general subsystem POVM gives measurement-dependent outcome distributions, leading to a minimum-error discrimination or accessible-information problem. In this case, both the accessible signal and its fluctuations depend on the charge algebra and the chosen measurement, and no universal crossover follows from the conserved-charge expectation values alone.

\subsection{Basis-Dependent Dephasing}

\begin{figure}    \includegraphics[width=\textwidth]{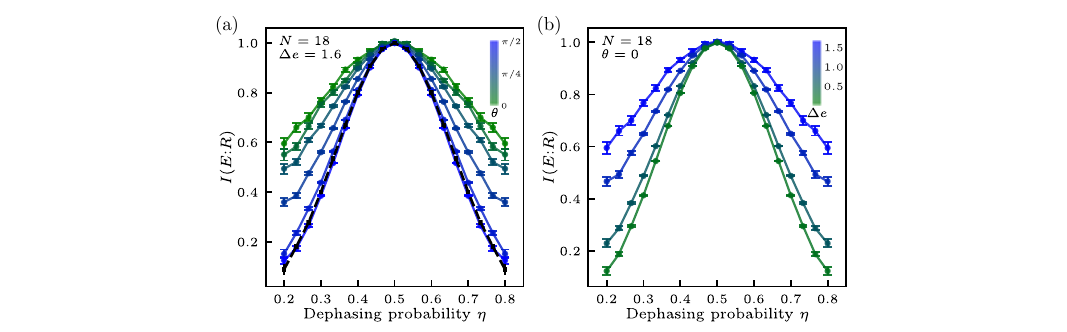}[h!]
  \caption{Basis dependence of information leakage. (a) $I(E:R)$ versus Pauli-error probability $\eta$ for mixed-field Ising codewords at $N=18$ and fixed energy-density separation $\Delta e=1.6$, with the dephasing axis varied from $\sigma^y$ $(\theta=0)$ to $\sigma^z$ $(\theta=\pi/2)$. The black dashed curve shows the equal-energy baseline, $\Delta e=0$. (b) $I(E:R)$ versus $\Delta e$ for $\sigma^y$-aligned dephasing $(\theta=0)$.}
  \label{fig:supp_basis_dep_deph}
    
\end{figure}

The preceding subsections highlight that differences in conserved quantities, in certain regimes, can be a mechanism for information leakage about the initial logical state to the environment $E$. However, here we show that the leakage must also be dependent on the noise channel, as well. Intuitively, if the channel implemented by the environment -- the complementary channel to the noise channel -- is insensitive to the conserved quantity, then the environment learns nothing about it, and the code remains error correcting irrespective of the magnitude of codeword charge differences.

To probe this dependence, we use the same decoupling setup, as described in Section~\ref{sec:aqec_decoupling}, for the two-dimensional codespace and apply an independent single-site dephasing channel whose axis is rotated in the $y$-$z$ plane,

\begin{equation} \label{eq:dephasing-y-z-rot}
    \mathcal{N}_{\theta, \eta}(\rho)
    = (1-\eta)\rho + \eta\,\sigma_\theta\,\rho\,\sigma_\theta,
    \qquad
    \sigma_\theta = \cos(\theta)\,\sigma_y + \sin(\theta)\,\sigma_z.
\end{equation}
Here, $\eta \in [0,1]$ is the probability of applying a Pauli operator $\sigma_\theta$, while $\theta$ determines its axis. The channel suppresses coherences in the $\sigma_\theta$ eigenbasis by a factor of $1-2\eta$, and in particular $\eta = 1/2$ corresponds to full dephasing. The full $N$-qubit noise channel is $\mathcal{N}^{(N)}_{\theta,\eta}
=
\bigotimes_{i=1}^N
\mathcal{N}^{(i)}_{\theta,\eta}$.A Stinespring dilation for the single-site channel is

\begin{equation}
    U_{\theta,\eta} ^{(i)} :=
\sqrt{1-\eta}\;
I_i\otimes\ket{0}_{E_i}
+
\sqrt{\eta}
\sigma_\theta^{(i)}\otimes\ket{1}_{E_i},
\end{equation}

and the product-channel dilation is $U_{\theta,\eta}^{(N)}
:=
\bigotimes_{i=1}^{N}
U_{\theta,\eta}^{(i)}$. Tracing out the environment reproduces the noise channel $\mathcal N_{\theta,\eta}^{(N)}(\rho)
=
\tr_E
\big(
U_{\theta,\eta}^{(N)}
\rho
U_{\theta,\eta}^{(N)\dagger}
\big)$. The environment coherently records the pattern of applied $\sigma_\theta$ operators. We quantify information leakage by the mutual information $I(E:R)$ obtained by applying $\operatorname{id}_R\otimes U_{\theta,\eta}^{(N)}$ to the encoded state $\sigma_{RS}(t)$ defined in Eq.~\eqref{eq:decoup_enc_pur_init}.

We evaluate this protocol for the mixed-field Ising model introduced above in Section~\ref{sec:numerical_therm_hamiltonians}. The Hamiltonian
contains the uniform field contribution $g\sum_{i=1}^{N}\sigma_i^y$, so a $\sigma^y$-aligned channel is directly sensitive to this
local component of the energy density. Varying the dephasing axis $\sigma_\theta$ in Eq.~\eqref{eq:dephasing-y-z-rot} therefore allows us to tune the sensitivity of the complementary channel to the local energetic structure distinguishing the two codewords. 

As shown in Fig.~\ref{fig:supp_basis_dep_deph}(a), $I(E:R)$ is largest near $\theta = 0$, where the dephasing axis is aligned with $\sigma^y$. By contrast, for $\theta = \pi/2$, corresponding to $\sigma^z$-aligned dephasing, $I(E:R)$ is indistinguishable from the equal-energy baseline, $\Delta e = 0$. Thus, even with large energy separation, a noise channel that is insensitive to the relevant local energetic component does not reveal this distinction to the environment. 

We further test this by varying the energy-density separation between the codewords. For $\theta=0$, where the complementary channel is sensitive to the dominant energy contrast, $I(E:R)$ increases with $\Delta e$, as shown in Fig.~\ref{fig:supp_basis_dep_deph}(b). Larger energy separation therefore produces more distinguishable environmental states and makes more information about the logical label accessible to the environment. These results demonstrate that conserved-quantity differences generate leakage only when the complementary channel resolves the degrees of freedom exhibiting those differences.

\bibliography{references}